\documentclass[12pt]{article}

\usepackage{epsfig,amsmath,amssymb,amsfonts,amstext,amsthm}
\usepackage{latexsym,graphics,epsf,epsfig,color}
\usepackage{cite,url}

\newtheorem{theorem}{Theorem}
\newtheorem{prop}{Proposition}
\newtheorem{lemma}{Lemma}
\newtheorem{coro}[lemma]{Corollary}

\begin{document}

\title{Symmetrical Multilevel Diversity Coding and \\Subset Entropy Inequalities}

\author{Jinjing~Jiang, Neeharika Marukala, and Tie Liu
\thanks{This research was supported in part by the National Science Foundation under Grant CCF-08-45848 and CCF-09-16867. The material in this paper was presented in part at the 2012 International Symposium on Network Coding (NetCod), Cambridge, MA, June 2012 and in part at the IEEE International Symposium on Information Theory (ISIT), Cambridge, MA, July 2012. J.~Jiang and T.~Liu are with the Department of Electrical and Computer Engineering, Texas A\&M University, College Station, TX 77843, USA (email: \{jinjing,tieliu\}@tamu.edu). N.~Marukala was with the Texas A\&M University. She is now with Qualcomm Incoporation, San Diego, CA 92121, USA (email: neeha.iitm@gmail.com).}}

\date{\today}

\maketitle

\begin{abstract}
Symmetrical multilevel diversity coding (SMDC) is a classical model for coding over distributed storage. In this setting, a simple separate encoding strategy known as superposition coding was shown to be optimal in terms of achieving the minimum sum rate (Roche, Yeung, and Hau, 1997) and the entire admissible rate region (Yeung and Zhang, 1999) of the problem. The proofs utilized carefully constructed induction arguments, for which the classical subset entropy inequality of Han (1978) played a key role. This paper includes two parts. In the first part the existing optimality proofs for classical SMDC are revisited, with a focus on their connections to subset entropy inequalities. First, a new sliding-window subset entropy inequality is introduced and then used to establish the optimality of superposition coding for achieving the minimum sum rate under a weaker source-reconstruction requirement. Second, a subset entropy inequality recently proved by Madiman and Tetali (2010) is used to develop a new structural understanding to the proof of Yeung and Zhang on the optimality of superposition coding for achieving the entire admissible rate region. Building on the connections between classical SMDC and the subset entropy inequalities developed in the first part, in the second part the optimality of superposition coding is further extended to the cases where there is either an additional all-access encoder (SMDC-A) or an additional secrecy constraint (S-SMDC).
\end{abstract}

\section{Introduction}\label{sec:Intr}
Symmetrical multilevel diversity coding (SMDC) is a classical model for coding over distributed storage, which was first introduced by Roche \cite{Roc-Thesis92} and Yeung \cite{Yeu-IT95}. In this setting, there are a total of $L$ \emph{independent} discrete memoryless sources $S_1,\ldots,S_L$, where the importance of the source $S_l$ is assumed to decrease with the subscript $l$. The sources are to be encoded by a total of $L$ \emph{randomly accessible} encoders. The goal of encoding is to ensure that the number of sources that can be nearly perfectly reconstructed grows with the number of available encoder outputs at the decoder. More specifically, denote by $U \subseteq \Omega_L:=\{1,\ldots,L\}$ the set of accessible encoders. The realization of $U$ is \emph{unknown} a priori at the encoders. However, the sources $S_1,\ldots,S_\alpha$ need to be nearly perfectly reconstructed whenever $|U| \geq \alpha$ at the decoder. The word ``symmetrical" here refers to the fact that the sources that need to be nearly perfectly reconstructed depend on the set of accessible encoders only via its cardinality. The rate allocations at different encoders, however, can be different and are not necessarily symmetrical.

A natural strategy for SMDC is to encode the sources separately at each of the encoders (no coding across different sources) known as \emph{superposition coding} \cite{Yeu-IT95}. To show that the natural superposition coding strategy is also optimal, however, turned out to be rather nontrivial. The optimality of superposition coding in terms of achieving the \emph{minimum sum rate} was established by Roche, Yeung, and Hau \cite{RYH-IT97}. The proof used a carefully constructed induction argument, for which the classical subset entropy inequality of Han \cite{Han-IC78} played a key role. Later, the optimality of superposition coding in terms of achieving the \emph{entire admission rate region} was established by Yeung and Zhang \cite{YZ-IT99}. Their proof was based on a \emph{new} subset entropy inequality, which was established by carefully combining Han's subset inequality with several highly technical results on the analysis of a sequence of linear programs (which are used to characterize the performance of superposition coding). 

This paper includes two parts. In the first part (Section~\ref{sec:SMDC}), the optimality proofs of \cite{RYH-IT97} and \cite{YZ-IT99} are revisited in light of two new subset entropy inequalities: 
\begin{itemize}
\item First, a new \emph{sliding-window} subset entropy inequality is introduced, which not only implies the classical subset entropy inequality of Han \cite{Han-IC78} in a trivial way, but also leads to a new proof of the optimality of superposition encoding for achieving the minimum sum rate under a \emph{weaker} source-reconstruction requirement. 
\item Second, a subset entropy inequality recently proved by Madiman and Tetali \cite{MT-IT10} is leveraged to provide a new \emph{structural} understanding to the subset entropy inequality of Yeung and Zhang \cite{YZ-IT99}. Based on this new understanding, a \emph{conditional} version of the subset entropy inequality of Yeung and Zhang \cite{YZ-IT99} is further established, which plays a key role in extending the optimality of superposition coding to the case where there is an additional secrecy constraint.
\end{itemize}

In the second part of the paper (Section~\ref{sec:Ext}), two extensions of classical SMDC are considered:
\begin{itemize}
\item The first extension, which we shall refer to as SMDC-A, features an \emph{all-access encoder}, in addition to the $L$ randomly accessible encoders in the classical setting, whose output is available at the decoder \emph{at all time}. This model is mainly motivated by the proliferation of mobile computing devices (laptop computers, tablets, smart phones etc.), which can access both remote storage nodes via unreliable wireless links and local hard disks which are always available but are of limited capacity. It is shown that in this setting, superposition coding remains optimal in terms of achieving the entire admissible rate region. Key to our proof is to identify the supporting hyperplanes that define the superposition coding rate region and then apply the subset entropy inequality of Yeung and Zhang \cite{YZ-IT99}. 
\item The second extension, which we shall refer to as S-SMDC, extends the problem of SMDC to the \emph{secure} communication setting. The problem was first introduced in \cite{BLLLM-ITR12}, where the optimality of superposition coding for achieving the minimum sum rate was established via the classical subset entropy inequality of Han \cite{Han-IC78}. Through the conditional version of the subset entropy inequality of Yeung and Zhang \cite{YZ-IT99} established in the first part, here we show that superposition coding can, in fact, achieve the entire admissible rate region of the problem, resolving the conjecture of \cite{BLLLM-ITR12} by positive.
\end{itemize}

\section{SMDC Revisited}\label{sec:SMDC}
\subsection{Problem Statement and Optimality of Superposition Coding}
\subsubsection{Problem Statement}
\begin{figure}[!t]
\centering
\includegraphics[width=0.95\linewidth,draft=false]{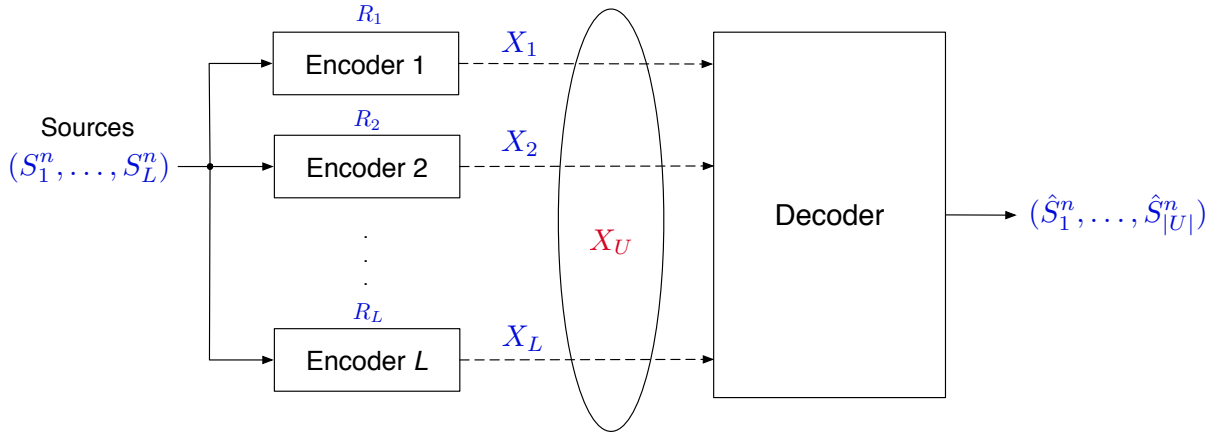}
\caption{The classical SMDC problem where a total of $L$ independent discrete memoryless sources $S_1,\ldots,S_L$ are to be encoded by a total of $L$ encoders. The decoder, which has access to a subset $U$ of the encoder outputs, needs to nearly perfectly reconstruct the sources $S_1,\ldots,S_{|U|}$ no matter what the realization of $U$ is.}
\label{fig:SMDC}
\end{figure}

As illustrated in Figure~\ref{fig:SMDC}, the problem of SMDC consists of:
\begin{itemize}
\item a total of $L$ \emph{independent} discrete memoryless sources $\{S_\alpha[t]\}_{t=1}^{\infty}$, where $\alpha=1,\ldots,L$ and $t$ is the time index;
\item a set of $L$ encoders (encoder $1$ to $L$);
\item a decoder which can access a nonempty subset $U \subseteq \Omega_L$ of the encoder outputs.
\end{itemize}
The realization of $U$ is \emph{unknown} a priori at the encoders. However, no matter which $U$ actually materializes, the decoder needs to nearly perfectly reconstruct the sources $S_1,\ldots,S_\alpha$ whenever $|U| \geq \alpha$. 

Formally, an $(n,(M_1,\ldots,M_L))$ code is defined by a collection of $L$ encoding functions:
\begin{equation}
e_l: \prod_{\alpha=1}^{L}\mathcal{S}_\alpha^n \rightarrow \{1,\ldots,M_l\}, \quad \forall l=1,\ldots,L
\end{equation}
and $2^L-1$ decoding functions:
\begin{equation}
d_U: \prod_{l\in U}\{1,\ldots,M_l\} \rightarrow \prod_{\alpha=1}^{|U|}\mathcal{S}_\alpha^n, \quad \forall U \subseteq \Omega_L \; \mbox{s.t.} \; U \neq \emptyset.
\end{equation}
A nonnegative rate tuple $(R_1,\ldots,R_L)$ is said to be \emph{admissible} if for every $\epsilon>0$, there exits, for sufficiently large block-length $n$, an $(n,(M_1,\ldots,M_L))$ code such that:
\begin{itemize}
\item (Rate constraints at the encoders)
\begin{equation}
\frac{1}{n}\log M_l \leq R_l +\epsilon, \qquad \forall l =1,\ldots,L;\label{eq:Rate}
\end{equation}
\item (Asymptotically perfect reconstructions at the decoder)
\begin{equation}
\mathrm{Pr}\left\{d_U(X_U) \neq (S_1^n,\ldots,S_{|U|}^n)\right\} \leq \epsilon, 
\qquad \forall U \subseteq \Omega_L \; \mbox{s.t.} \; U \neq \emptyset
\label{eq:PR}
\end{equation}
where $S_\alpha^n := \{S_\alpha[t]\}_{t=1}^{n}$, $X_l:=e_l(S_1^n,\ldots,S_L^n)$ is the output of encoder $l$, and $X_U:=\{X_l: l \in U\}$.
\end{itemize}
The \emph{admissible rate region} $\mathcal{R}$ is the collection of \emph{all} admissible rate tuples $(R_1,\ldots,R_L)$. The \emph{minimum sum rate} $R_{ms}$ is defined as
\begin{equation}
R_{ms} := \min_{(R_1,\ldots,R_L) \in \mathcal{R}}\sum_{l=1}^LR_l.
\end{equation}

\subsubsection{Superposition Coding Rate Region}
As mentioned previously, a natural strategy for SMDC is \emph{superposition coding}, i.e., to encode the sources separately at the encoders and there is \emph{no} coding across different sources. Formally, the problem of encoding a single source $S_\alpha$ can be viewed as a special case of the general SMDC problem, where the sources $S_m$ are deterministic for all $m \neq \alpha$. In this case, the source $S_\alpha$ needs to be nearly perfectly reconstructed whenever the decoder can access at least $\alpha$ encoder outputs. Thus, the problem is essentially to transmit $S_\alpha$ over an \emph{erasure} channel, and the following simple source-channel separation scheme is known to be optimal \cite{Roc-Thesis92,Yeu-IT95}:

\begin{itemize}
\item First compress the source sequence $S_\alpha^n$ into a source message $W_\alpha$ using a \emph{lossless} source code. It is well known \cite[Ch.~5]{CT-B06} that the rate of the source message $W_\alpha$ can be made arbitrarily close to the entropy rate $H(S_\alpha)$ for sufficiently large block-length $n$.
\item Next, the source message $W_\alpha$ is encoded at encoders $1$ to $L$ using a \emph{maximum distance separable} code \cite{Sin-IT64}. It is well known \cite{Roc-Thesis92,Yeu-IT95} that the source message $W_\alpha$ can be perfectly recovered at the decoder whenever 
\begin{equation}
\sum_{l \in U} R_l \geq \frac{1}{n}H(W_\alpha), \quad \forall U \in \Omega_L^{(\alpha)}
\end{equation}
for sufficiently large block length $n$, where $\Omega_L^{(\alpha)}$ denotes the collection of all subsets of $\Omega_L$ of size $\alpha$.
\end{itemize}

Combining the above two steps, we conclude that the admissible rate region for encoding a single source $S_\alpha$ is given by the collection of all nonnegative rate tuples $(R_1,\ldots,R_L)$ satisfying 
\begin{equation}
\sum_{l \in U} R_l \geq H(S_\alpha), \quad \forall U \in \Omega_L^{(\alpha)}.
\label{eq:SSDC}
\end{equation}
By definition, the superposition coding rate region $\mathcal{R}_{sup}$ for encoding the sources $S_1,\ldots,S_L$ is given by the collection of all nonnegative rate tuples $(R_1,\ldots,R_L)$ such that
\begin{equation}
R_l := \sum_{\alpha=1}^{L}r_l^{(\alpha)}
\label{eq:SMDC_sup}
\end{equation}
for some nonnegative $r_l^{(\alpha)}$, $\alpha=1,\ldots,L$ and $l=1,\ldots,L$, satisfying  
\begin{equation}
\sum_{l \in U} r_l^{(\alpha)} \geq H(S_\alpha), \quad \forall U \in \Omega_L^{(\alpha)}.
\label{eq:SSDC2}
\end{equation}

In principle, an explicit characterization of the superposition coding rate region $\mathcal{R}_{sup}$ can be obtained by eliminating $r_l^{(\alpha)}$, $\alpha=1,\ldots,L$ and $l=1,\ldots,L$, via a Fourier-Motzkin elimination from \eqref{eq:SMDC_sup} and \eqref{eq:SSDC2}. However, the elimination process is \emph{unmanageable} even for moderate $L$, as there are simply too many equations involved. On the other hand, note that the superposition coding rate region $\mathcal{R}_{sup}$ is a convex polyhedron with polyhedral cone being $(\mathbb{R}^+)^L$, so an equivalent characterization is to characterize the supporting hyperplanes: 
\begin{equation}
\sum_{l=1}^L\lambda_lR_l \geq f(\boldsymbol{\lambda}), \quad \forall \boldsymbol{\lambda}:=(\lambda_1,\ldots,\lambda_L) \in (\mathbb{R}^+)^L
\label{eq:sup}
\end{equation}
where
\begin{eqnarray}
f(\boldsymbol{\lambda}) & = & \min_{(R_1,\ldots,R_L) \in \mathcal{R}_{sup}}\sum_{l=1}^L\lambda_lR_l\\
& = & 
\begin{array}{rl}
\min & \sum_{l=1}^L\left(\sum_{\alpha=1}^L\lambda_lr_l^{(\alpha)}\right)\\
\mbox{subject to} & \sum_{l \in U} r_l^{(\alpha)} \geq H(S_\alpha), \quad \forall U \in \Omega_L^{(\alpha)} \; \mbox{and} \; \alpha=1,\ldots,L\\
& r_l^{(\alpha)} \geq 0, \quad \forall \alpha=1,\ldots,L \; \mbox{and} \; l=1,\ldots,L.
\end{array}
\end{eqnarray}
Clearly, the above optimization problem can be separated into the following $L$ sub-optimization problems:
\begin{equation}
f(\boldsymbol{\lambda})=\sum_{\alpha=1}^Lf'_\alpha(\boldsymbol{\lambda})
\end{equation}
where
\begin{eqnarray}
f'_\alpha(\boldsymbol{\lambda}) & = & 
\begin{array}{rl}
\min & \sum_{l=1}^L\lambda_lr_l^{(\alpha)}\\
\mbox{subject to} & \sum_{l \in U} r_l^{(\alpha)} \geq H(S_\alpha), \quad \forall U \in \Omega_L^{(\alpha)}\\
& r_l^{(\alpha)} \geq 0, \quad \forall l=1,\ldots,L
\end{array}\\
& = & 
\begin{array}{rl}
\max & \left(\sum_{U \in \Omega_L^{(\alpha)}}c_{\boldsymbol{\lambda}}(U)\right)H(S_\alpha)\\
\mbox{subject to} & \sum_{\{U \in \Omega_L^{(\alpha)}: U \ni l\}} c_{\boldsymbol{\lambda}}(U) \leq \lambda_l, \quad \forall l=1,\ldots,L\\
& c_{\boldsymbol{\lambda}}(U) \geq 0, \quad \forall U \in \Omega_L^{(\alpha)}.
\end{array}
\label{eq:dual}
\end{eqnarray}
and \eqref{eq:dual} follows from the strong \emph{duality} for linear programs. For any $\boldsymbol{\lambda} \in (\mathbb{R}^+)^L$ and any $\alpha=1,\ldots,L$, let
\begin{eqnarray}
f_\alpha(\boldsymbol{\lambda}) & := & 
\begin{array}{rl}
\max & \sum_{U \in \Omega_L^{(\alpha)}}c_{\boldsymbol{\lambda}}(U)\\
\mbox{subject to} & \sum_{\{U \in \Omega_L^{(\alpha)}: U \ni l\}} c_{\boldsymbol{\lambda}}(U) \leq \lambda_l, \quad \forall l=1,\ldots,L\\
& c_{\boldsymbol{\lambda}}(U) \geq 0, \quad \forall U \in \Omega_L^{(\alpha)}.
\end{array}
\label{eq:dual2}
\end{eqnarray}
Then, we have $f'_\alpha(\boldsymbol{\lambda})=f_\alpha(\boldsymbol{\lambda})H(S_\alpha)$ and hence 
\begin{equation}
f(\boldsymbol{\lambda})=\sum_{\alpha=1}^Lf_\alpha(\boldsymbol{\lambda})H(S_\alpha)
\label{eq:sup2}
\end{equation}
for any $\boldsymbol{\lambda} \in (\mathbb{R}^+)^L$. Substituting \eqref{eq:sup2} into \eqref{eq:sup}, we conclude that the superposition coding rate region $\mathcal{R}_{sup}$ is given by the collection of nonnegative rate tuples $(R_1,\ldots,R_L)$ satisfying
\begin{equation}
\sum_{l=1}^L\lambda_lR_l \geq \sum_{\alpha=1}^Lf_\alpha(\boldsymbol{\lambda})H(S_\alpha), \quad \forall \boldsymbol{\lambda} \in (\mathbb{R}^+)^L.
\label{eq:sup3}
\end{equation}

For a general $\boldsymbol{\lambda}$, the linear program \eqref{eq:dual2} does not admit a \emph{closed-form} solution. However, for $\boldsymbol{\lambda}=\boldsymbol{1}:=(1,\ldots,1)$ it can be easily verified that $c_{\boldsymbol{1}}^{(\alpha)}=\{c_{\boldsymbol{1}}(U): U \in \Omega_L^{(\alpha)}\}$ where
\begin{equation}
c_{\boldsymbol{1}}(U):=\frac{1}{
\left(
\begin{array}{c}
  L-1   \\
  \alpha-1   
\end{array}
\right)
}
\label{eq:C1u}
\end{equation}
is an \emph{optimal} solution to the linear program \eqref{eq:dual2}, and we thus have
\begin{equation}
f_\alpha(\boldsymbol{1}) = \sum_{U \in \Omega_L^{(\alpha)}}c_{\boldsymbol{1}}(U)
=\frac{
\left(
\begin{array}{c}
  L   \\
  \alpha   
\end{array}
\right)
}
{
\left(
\begin{array}{c}
  L-1   \\
  \alpha-1   
\end{array}
\right)
}
=\frac{L}{\alpha}
\end{equation}
for any $\alpha=1,\ldots,L$. Hence, the minimum sum rate that can be achieved by superposition coding is given by
\begin{equation}
\min_{(R_1,\ldots,R_L) \in \mathcal{R}_{sup}}\sum_{l=1}^LR_l =f(\boldsymbol{1})=\sum_{\alpha=1}^Lf_\alpha(\boldsymbol{1})H(S_\alpha)=\sum_{\alpha=1}^L(L/\alpha)H(S_\alpha).
\end{equation}

\subsubsection{Optimality of Superposition Coding: Known Proofs}
To show that superposition coding is optimal in terms of achieving the entire admissible rate region, we need to show that 
for any $\boldsymbol{\lambda} \in (\mathbb{R}^+)^L$ we have
\begin{equation}
\sum_{l=1}^L\lambda_lR_l \geq \sum_{\alpha=1}^Lf_\alpha(\boldsymbol{\lambda})H(S_\alpha), \quad \forall (R_1,\ldots,R_L) \in \mathcal{R}.
\label{eq:RR-Conv}
\end{equation}
In particular, to show that superposition coding is optimal in terms of achieving the minimum sum rate, we need to show that
\begin{equation}
\sum_{l=1}^LR_l \geq \sum_{\alpha=1}^Lf_\alpha(\boldsymbol{1})H(S_\alpha) = \sum_{\alpha=1}^L(L/\alpha)H(S_\alpha), \quad \forall (R_1,\ldots,R_L) \in \mathcal{R}.
\label{eq:SR-Conv}
\end{equation}

Note that for any admissible rate tuple $(R_1,\ldots,R_L) \in \mathcal{R}$ and $\epsilon>0$, by the rate constraints \eqref{eq:Rate} we have
\begin{equation}
n(R_l+\epsilon) \geq H(X_l), \quad \forall l=1,\ldots,L
\label{eq:Rate2}
\end{equation}
for sufficiently large block-length $n$. Furthermore, by the asymptotically perfect reconstruction requirement \eqref{eq:PR} and the well-known Fano's inequality we have
\begin{equation}
H(S_1^n,\ldots,S_\alpha^n|X_U) \leq n\delta_{\alpha}^{(n)}
\end{equation}
for any $U \in \Omega_L^{(\alpha)}$ and $\alpha=1,\ldots,L$, where $\delta_\alpha^{(n)} \rightarrow 0$ in the limit as $n \rightarrow \infty$ and $\epsilon \rightarrow 0$. Thus, for any $V \in \Omega_L^{(\alpha-1)}$ we have
\begin{align}
H(X_V|S_1^n,\ldots,S_{\alpha-2}^n) &= H(X_V|S_1^n,\ldots,S_{\alpha-1}^n)+
I(X_V;S_{\alpha-1}^n|S_1^n,\ldots,S_{\alpha-2}^n)\\
&=  H(X_V|S_1^n,\ldots,S_{\alpha-1}^n)+
H(S_{\alpha-1}^n|S_1^n,\ldots,S_{\alpha-2}^n)-\nonumber\\
& \hspace{15pt}H(S_{\alpha-1}^n|S_1^n,\ldots,S_{\alpha-2}^n,X_V)\\
& \geq H(X_V|S_1^n,\ldots,S_{\alpha-1}^n)+
H(S_{\alpha-1}^n)-H(S_1^n,\ldots,S_{\alpha-1}^n|X_V)\label{eq:SMDC-T1}\\
& \geq H(X_V|S_1^n,\ldots,S_{\alpha-1}^n)+
nH(S_{\alpha-1})-n\delta_{\alpha-1}^{(n)}\label{eq:PR2}
\end{align}
where \eqref{eq:SMDC-T1} follows from the facts that all sources are independent so $H(S_{\alpha-1}^n|S_1^n,\ldots,S_{\alpha-2}^n)=H(S_{\alpha-1}^n)$ and that 
\begin{eqnarray}
H(S_{\alpha-1}^n|S_1^n,\ldots,S_{\alpha-2}^n,X_V) &=& H(S_1^n,\ldots,S_{\alpha-1}^n|X_V)-H(S_1^n,\ldots,S_{\alpha-2}^n|X_V)\\
& \leq & H(S_1^n,\ldots,S_{\alpha-1}^n|X_V).
\end{eqnarray}
Therefore, starting with \eqref{eq:Rate2} and applying \eqref{eq:PR2} \emph{iteratively} may lead us towards a proof of \eqref{eq:RR-Conv} and \eqref{eq:SR-Conv}. Note, however, that to apply \eqref{eq:PR2} iteratively we shall need to bound from below $H(X_V|S_1^n,\ldots,S_{\alpha-1}^n)$ in terms of $H(X_U|S_1^n,\ldots,S_{\alpha-1}^n)$ for some $U \in \Omega_L^{(\alpha)}$. The key observation of \cite{RYH-IT97} and \cite{YZ-IT99} is that such bounds exist, not for an arbitrary individual pair of $U$ and $V$, but rather at the level of an appropriate \emph{averaging} among $V \in \Omega_L^{(\alpha-1)}$ and $U \in \Omega_L^{(\alpha)}$.

More specifically, \cite{RYH-IT97} considered the classical subset entropy inequality of Han \cite{Han-IC78}, which can be written as follows.

\begin{theorem}[A subset entropy inequality of Han \cite{Han-IC78}]\label{thm:Han}
For any collection of $L$ jointly distributed random variables $(X_1,\ldots,X_L)$, we have
\begin{equation}
\frac{1}
{
\left(
\begin{array}{c}
  L   \\
  \alpha-1   
\end{array}
\right)
}\sum_{V \in \Omega_L^{(\alpha-1)}}\frac{H(X_V)}{\alpha-1} \geq 
\frac{1}
{
\left(
\begin{array}{c}
  L   \\
  \alpha   
\end{array}
\right)
}\sum_{U \in \Omega_L^{(\alpha)}}\frac{H(X_U)}{\alpha}
\label{eq:Han}
\end{equation}
for any $\alpha=2,\ldots,L$.
\end{theorem}

Iteratively applying \eqref{eq:PR2} and \eqref{eq:Han}, we may obtain
\begin{align}
\frac{1}{L}
\sum_{l=1}^LH(X_l) & = \frac{1}
{
\left(
\begin{array}{c}
  L   \\
  1   
\end{array}
\right)}
\sum_{V \in \Omega_L^{(1)}}H(X_V)\\
& \geq
\frac{1}
{
\left(
\begin{array}{c}
  L   \\
  m   
\end{array}
\right)}
\sum_{U \in \Omega_L^{(m)}}\frac{H(X_U|S_1^n,\ldots,S_m^n)}{m}+
n\sum_{\alpha=1}^m\frac{H(S_\alpha)}{\alpha}-n\sum_{\alpha=1}^m\frac{\delta_\alpha^{(n)}}{\alpha}
\end{align}
for any $m=1,\ldots,L$. In particular, let $m=L$, and we have
\begin{eqnarray}
\frac{1}{L}
\sum_{l=1}^LH(X_l) & \geq &
\frac{1}
{
\left(
\begin{array}{c}
  L   \\
  L   
\end{array}
\right)}
\sum_{U \in \Omega_L^{(L)}}\frac{H(X_U|S_1^n,\ldots,S_L^n)}{L}+
n\sum_{\alpha=1}^L\frac{H(S_\alpha)}{\alpha}-n\sum_{\alpha=1}^L\frac{\delta_\alpha^{(n)}}{\alpha}\\
& \geq & n\sum_{\alpha=1}^L\frac{H(S_\alpha)}{\alpha}-n\sum_{\alpha=1}^L\frac{\delta_\alpha^{(n)}}{\alpha}.\label{eq:SMDC-T2}
\end{eqnarray}
Substituting \eqref{eq:Rate2} into \eqref{eq:SMDC-T2} and dividing both sides of the inequality by $n$, we have
\begin{equation}
\frac{1}{L}\sum_{l=1}^L(R_l+\epsilon) \geq \sum_{\alpha=1}^L\frac{H(S_\alpha)}{\alpha}-\sum_{\alpha=1}^L\frac{\delta_\alpha^{(n)}}{\alpha}.
\end{equation}
Finally, letting $n \rightarrow \infty$ and $\epsilon \rightarrow 0$ completes the proof of \eqref{eq:SR-Conv}, i.e., superposition coding can achieve the minimum sum rate for the general SMDC problem.

To prove that superposition coding can in fact achieve the entire admissible rate region, Yeung and Zhang \cite{YZ-IT99} proved the following key subset entropy inequality.

\begin{theorem}[A subset entropy inequality of Yeung and Zhang \cite{YZ-IT99}]\label{thm:YZ}
For any $\boldsymbol{\lambda} \in (\mathbb{R}^+)^L$, there exists a function $c_{\boldsymbol{\lambda}}: 2^{\Omega_L}\setminus \emptyset \rightarrow \mathbb{R}^+$ such that:
\begin{itemize}
\item[1)] for each $\alpha=1,\ldots,L$, $c_{\boldsymbol{\lambda}}^{(\alpha)}:=\{c_{\boldsymbol{\lambda}}(U): U \in \Omega_L^{(\alpha)}\}$ is an \emph{optimal} solution to the linear program \eqref{eq:dual2}; and 
\item[2)] for each $\alpha=2,\ldots,L$,
\begin{equation}
\sum_{V \in \Omega_L^{(\alpha-1)}}c_{\boldsymbol{\lambda}}(V)H(X_V) \geq \sum_{U \in \Omega_L^{(\alpha)}}c_{\boldsymbol{\lambda}}(U)H(X_U)\label{eq:YZ}
\end{equation}
for any collection of $L$ jointly distributed random variables $(X_1,\ldots,X_L)$.
\end{itemize}
\end{theorem}

Iteratively applying \eqref{eq:PR2} and \eqref{eq:YZ}, we may obtain
\begin{align}
\sum_{V \in \Omega_L^{(1)}}c_{\boldsymbol{\lambda}}(V)H(V) \geq \sum_{U \in \Omega_L^{(m)}}c_{\boldsymbol{\lambda}}(U)H(X_U|S_1^n,\ldots,S_m^n)+
n\sum_{\alpha=1}^mf_\alpha(\boldsymbol{\lambda})H(S_\alpha)-n\sum_{\alpha=1}^mf_\alpha(\boldsymbol{\lambda})\delta_\alpha^{(n)}
\end{align}
for any $m=1,\ldots,L$. In particular, let $m=L$, and note that for $\alpha=1$ the optimal solution to the linear program \eqref{eq:dual2} is unique and is given by
\begin{equation}
c_{\boldsymbol{\lambda}}(\{l\})=\lambda_l, \quad \forall l \in \Omega_L.
\end{equation}
We have
\begin{eqnarray}
\hspace{-20pt} \sum_{l=1}^L\lambda_lH(X_l) & \geq &
\sum_{U \in \Omega_L^{(L)}}c_{\boldsymbol{\lambda}}(U)H(X_U|S_1^n,\ldots,S_L^n)+
n\sum_{\alpha=1}^Lf_\alpha(\boldsymbol{\lambda})H(S_\alpha)-n\sum_{\alpha=1}^Lf_\alpha(\boldsymbol{\lambda})\delta_\alpha^{(n)}\\
& \geq & n\sum_{\alpha=1}^Lf_\alpha(\boldsymbol{\lambda})H(S_\alpha)-n\sum_{\alpha=1}^Lf_\alpha(\boldsymbol{\lambda})\delta_\alpha^{(n)}.\label{eq:SMDC-T3}
\end{eqnarray}
Substituting \eqref{eq:Rate2} into \eqref{eq:SMDC-T3} and dividing both sides of the inequality by $n$, we have
\begin{equation}
\sum_{l=1}^L\lambda_l(R_l+\epsilon) \geq \sum_{\alpha=1}^Lf_\alpha(\boldsymbol{\lambda})H(S_\alpha)-\sum_{\alpha=1}^Lf_\alpha(\boldsymbol{\lambda})\delta_\alpha^{(n)}.
\end{equation}
Finally, letting $n \rightarrow \infty$ and $\epsilon \rightarrow 0$ completes the proof of \eqref{eq:RR-Conv}, i.e., superposition coding can achieve the entire admissible rate region for the general SMDC problem.

\subsection{Minimum Sum Rate via a Sliding-Window Subset Entropy Inequality}
In this section, we prove a new \emph{sliding-window} subset entropy inequality and then use it to provide an alternative proof of the optimality of superposition coding for achieving the minimum sum rate. 

\subsubsection{A sliding-window subset entropy inequality}\label{subsubsec:sw}
\begin{figure}[!t]
\centering
\includegraphics[width=0.6\linewidth,draft=false]{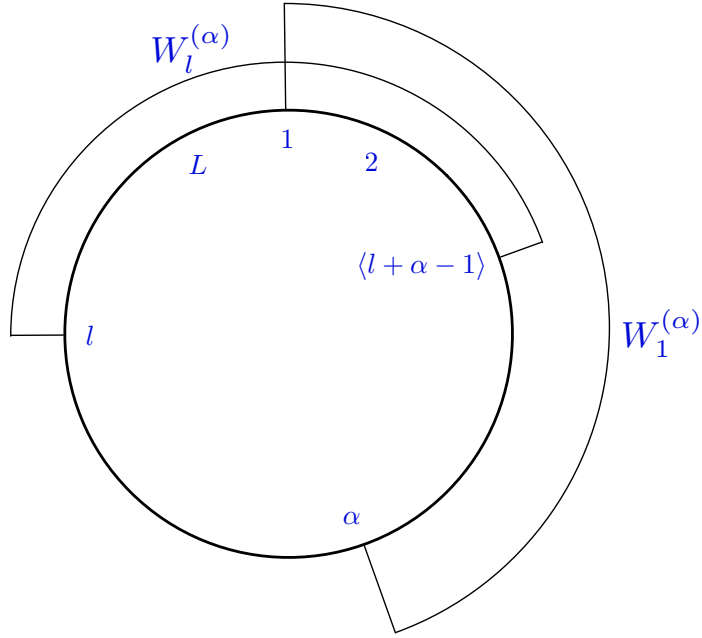}
\caption{An illustration of the sliding windows of length $\alpha$ when the integers $1,\ldots,L$ are circularly placed (clockwise) based on their natural order.}
\label{fig:sw}
\end{figure}

For any integer $l$ let
\begin{equation}
\langle l\rangle := \left\{
\begin{array}{ll}
l \bmod L, & \mbox{if $l \bmod L \neq 0$}\\
L, & \mbox{if $l \bmod L = 0$}
\end{array}
\right.
\end{equation} 
and for any $l=1,\ldots,L$ and $\alpha=1,\ldots,L$ let 
\begin{equation}
W_l^{(\alpha)} := \{l,\langle l+1\rangle,\ldots,\langle l+\alpha-1\rangle\}.
\end{equation}
As illustrated in Figure~\ref{fig:sw}, $W_l^{(\alpha)}$ represents a \emph{sliding window} of length $\alpha$ starting with $l$ when the integers $1,\ldots,L$ are circularly placed (clockwise or counter clockwise) based on their natural order. We have the following \emph{sliding-window} subset entropy inequality.

\begin{theorem}[A sliding-window subset entropy inequality] \label{thm:SW}
For any collection of $L$ jointly distributed random variables $(X_1,\ldots,X_L)$, we have
\begin{equation}
\sum_{l=1}^{L}\frac{H(X_{W_l^{(\alpha-1)}})}{\alpha-1} \geq \sum_{l=1}^{L}\frac{H(X_{W_l^{(\alpha)}})}{\alpha}
\label{eq:SW}
\end{equation}
for any $\alpha=2,\ldots,L$. The equalities hold when $X_1,\ldots,X_L$ are mutually independent of each other.
\end{theorem}

\begin{proof}
Consider a proof via an induction on $\alpha$. First, for $\alpha=2$ we have
\begin{eqnarray}
\sum_{l=1}^{L}H(X_{W_l^{(1)}}) &=& \sum_{l=1}^{L}H(X_l)\\
& = & \sum_{l=1}^L\frac{H(X_l)+H(X_{\langle l+1\rangle})}{2}\\
& \geq & \sum_{l=1}^L\frac{H(X_l,X_{\langle l+1\rangle})}{2} \label{L4-1}\\
&=& \sum_{l=1}^{L}\frac{H(X_{W_l^{(2)}})}{2}
\end{eqnarray}
where \eqref{L4-1} follows from the independence bound on entropy.

Next, assume that the inequality \eqref{eq:SW} holds for $\alpha=r$ for some $r \in \{2,\ldots,L-1\}$, i.e.,
\begin{align}
\sum_{l=1}^{L}\frac{H(X_{W_l^{(r-1)}})}{r-1} \geq \sum_{l=1}^{L}\frac{H(X_{W_l^{(r)}})}{r}.\label{eq:swInd}
\end{align}
We have 
\begin{eqnarray}
\sum_{l=1}^{L}H(X_{W_l^{(r)}}) &=& \frac{1}{2}\sum_{l=1}^{L}\left[H(X_{W_l^{(r)}})+H(X_{W_{\langle l+1\rangle}^{(r)}})\right]\\
& \geq & \frac{1}{2}\sum_{l=1}^{L}\left[H(X_{W_l^{(r+1)}})+H(X_{W_{\langle l+1\rangle}^{(r-1)}})\right]\label{L4-2}\\
& = & \frac{1}{2}\sum_{l=1}^{L}H(X_{W_l^{(r+1)}})+\frac{1}{2}\sum_{l=1}^{L}H(X_{W_{\langle l+1\rangle}^{(r-1)}})\\
& = & \frac{1}{2}\sum_{l=1}^{L}H(X_{W_l^{(r+1)}})+\frac{1}{2}\sum_{l=1}^{L}H(X_{W_l^{(r-1)}})\label{L4-3}\\
& \geq & \frac{1}{2}\sum_{l=1}^{L}H(X_{W_l^{(r+1)}})+\frac{1}{2}\cdot\frac{r-1}{r}\sum_{l=1}^{L}H(X_{W_l^{(r)}})\label{L4-4}
\end{eqnarray}
where \eqref{L4-2} follows from the submodularity of entropy \cite[Ch.~14.A]{Yeu-B08}
\begin{equation}
H(X_U)+H(X_V) \geq H(X_{U \cup V})+H(X_{U \cap V})
\end{equation}
for $U=W_l^{(r)}$ and $V=W_{\langle l+1\rangle}^{(r)}$ so $U \cup V =W_l^{(r+1)}$ and $U \cap V=W_{\langle l+1\rangle}^{(r-1)}$, and \eqref{L4-4} follows from the induction assumption \eqref{eq:swInd}. Moving the second term on the right-hand side of \eqref{L4-4} to the left and multiplying both sides by $\frac{2}{r+1}$, we have
\begin{equation}
\frac{1}{r}\sum_{l=1}^{L}H(X_{W_l^{(r)}}) \geq \frac{1}{r+1}\sum_{l=1}^{L}H(X_{W_l^{(r+1)}}).
\end{equation}
We have thus proved that the inequality \eqref{eq:SW} also holds for $\alpha=r+1$.

Finally, note that when $X_1,\ldots,X_L$ are mutually independent, we have
\begin{equation} 
\sum_{l=1}^{L}\frac{H(X_{W_l^{(\alpha)}})}{\alpha} = \sum_{l=1}^{L}H(X_l), \quad \forall \alpha=1,\ldots,L.
\end{equation}
This completes the proof of Theorem~\ref{thm:SW}. 
\end{proof}

Note that for $\alpha=L$, the classical subset entropy inequality of Han \eqref{eq:Han} and the sliding-window subset entropy inequality \eqref{eq:SW} are equivalent, and both can be equivalently written as
\begin{align}
\frac{1}{L-1}\sum_{l=1}^L H(X_{\Omega_L\setminus\{l\}}) \geq H(X_{\Omega_L}).
\label{eq:Han2}
\end{align} 
For a \emph{general} $\alpha$, the classical subset entropy inequality of Han \eqref{eq:Han} can be derived from the sliding-window subset entropy inequality \eqref{eq:SW} via a simple \emph{permutation} argument as follows. Let $\pi$ be a permutation on $\Omega_L$. For any $l=1,\ldots,L$ and $\alpha=1,\ldots,L$, let
\begin{equation}
W_{\pi,l}^{(\alpha)} := \{\pi^{-1}(l),\pi^{-1}(\langle l+1\rangle),\ldots,\pi^{-1}(\langle l+\alpha-1\rangle)\}.
\end{equation}
By Theorem~\ref{thm:SW}, we have
\begin{equation}
\frac{1}{\alpha-1}\sum_{l=1}^{L}H(X_{W_{\pi,l}^{(\alpha-1)}}) \geq \frac{1}{\alpha}\sum_{l=1}^{L}H(X_{W_{\pi,l}^{(\alpha)}})
\label{eq:SW2}
\end{equation}
for any $\alpha=2,\ldots,L$. Averaging \eqref{eq:SW2} over all possible permutations $\pi$, we have
\begin{equation}
\frac{1}{L!}\sum_{\pi}\left[\frac{1}{\alpha-1}\sum_{l=1}^{L}H(X_{W_{\pi,l}^{(\alpha-1)}})\right] \geq \frac{1}{L!}\sum_{\pi}\left[\frac{1}{\alpha}\sum_{l=1}^{L}H(X_{W_{\pi,l}^{(\alpha)}})\right].
\label{eq:SW3}
\end{equation}
Note that for any $\alpha=1,\ldots,L$,
\begin{equation}
\sum_{\pi}\sum_{l=1}^{L}H(X_{W_{\pi,l}^{(\alpha)}}) = L\cdot\alpha!(L-\alpha)!\sum_{U \in \Omega_L^{(\alpha)}}H(X_U). 
\label{eq:SW4}
\end{equation}
Substituting \eqref{eq:SW4} into \eqref{eq:SW3} and dividing both sides of the inequality by $L$ establish the classical subset entropy inequality of Han \eqref{eq:Han}.

\subsubsection{The minimum sum rate}
The sliding-window subset entropy inequality \eqref{eq:SW} can be used to provide an alternative proof of the optimality of superposition coding for achieving the minimum sum rate as follows. Let us first show that 
\begin{align}
\frac{1}{L}\sum_{l=1}^LH(X_l) &= \frac{1}{L}\sum_{l=1}^LH(X_{W_l^{(1)}})\\
& \geq
\frac{1}{L}
\sum_{l=1}^L\frac{H(X_{W_l^{(m)}}|S_1^n,\ldots,S_m^n)}{m}+
n\sum_{\alpha=1}^m\frac{H(S_\alpha)}{\alpha}-n\sum_{\alpha=1}^m\frac{\delta_\alpha^{(n)}}{\alpha}
\label{eq:SMDC-T100}
\end{align}
for any $m=1,\ldots,L$. 

Consider a proof via an induction on $m$. When $m=1$, \eqref{eq:SMDC-T100} can be written as
\begin{equation}
\frac{1}{L}\sum_{l=1}^LH(X_l) \geq \frac{1}{L}\sum_{l=1}^LH(X_l|S_1^n)+nH(S_1)-n\delta_1^{(n)}
\end{equation}
which can be obtained via a uniform averaging of \eqref{eq:PR2} for $\alpha=2$ and $V=\{l\}$ for $l=1,\ldots,L$. Now assume that the inequality \eqref{eq:SMDC-T100} holds for $m=r-1$ for some $r \in \{2,\ldots,L\}$. We have
\begin{align}
\frac{1}{L}\sum_{l=1}^LH(X_l) & \geq \frac{1}{L}
\sum_{l=1}^L\frac{H(X_{W_l^{(r-1)}}|S_1^n,\ldots,S_{r-1}^n)}{r-1}+
n\sum_{\alpha=1}^{r-1}\frac{H(S_\alpha)}{\alpha}-n\sum_{\alpha=1}^{r-1}\frac{\delta_\alpha^{(n)}}{\alpha}\\
& \geq \frac{1}{L}
\sum_{l=1}^L\frac{H(X_{W_l^{(r)}}|S_1^n,\ldots,S_{r-1}^n)}{r}+
n\sum_{\alpha=1}^{r-1}\frac{H(S_\alpha)}{\alpha}-n\sum_{\alpha=1}^{r-1}\frac{\delta_\alpha^{(n)}}{\alpha}
\label{eq:SMDC-T200}
\end{align}
where \eqref{eq:SMDC-T200} follows from the sliding-window subset entropy inequality \eqref{eq:SW} with $\alpha=r$. Letting $\alpha=r+1$ and $V=W_l^{(r)}$ in \eqref{eq:PR2}, we have
\begin{align}
H(X_{W_l^{(r)}}|S_1^n,\ldots,S_{r-1}^n) & \geq H(X_{W_l^{(r)}}|S_1^n,\ldots,S_{r}^n)+nH(S_{r})-n\delta_{r}^{(n)}.
\label{eq:SMDC-T300}
\end{align}
Substituting \eqref{eq:SMDC-T300} into \eqref{eq:SMDC-T200} gives
\begin{align}
\frac{1}{L}\sum_{l=1}^LH(X_l)
& \geq
\frac{1}{L}
\sum_{l=1}^L\frac{H(X_{W_l^{(r)}}|S_1^n,\ldots,S_{r}^n)}{r}+
n\sum_{\alpha=1}^{r}\frac{H(S_\alpha)}{\alpha}-n\sum_{\alpha=1}^{r}\frac{\delta_\alpha^{(n)}}{\alpha}.
\end{align}
This completes the proof of the induction step and hence \eqref{eq:SMDC-T100}.

Now let $m=L$, and we have
\begin{align}
\frac{1}{L}\sum_{l=1}^LH(X_l) & \geq 
\frac{1}{L}
\sum_{l=1}^L\frac{H(X_{W_l^{(L)}}|S_1^n,\ldots,S_L^n)}{L}+
n\sum_{\alpha=1}^L\frac{H(S_\alpha)}{\alpha}-n\sum_{\alpha=1}^L\frac{\delta_\alpha^{(n)}}{\alpha}
\\
& \geq n\sum_{\alpha=1}^L\frac{H(S_\alpha)}{\alpha}-n\sum_{\alpha=1}^L\frac{\delta_\alpha^{(n)}}{\alpha}.
\label{eq:SMDC-T400}
\end{align}
Substituting \eqref{eq:Rate2} into \eqref{eq:SMDC-T400} and dividing both sides of the inequality by $n$, we have
\begin{equation}
\frac{1}{L}\sum_{l=1}^L(R_l+\epsilon) \geq \sum_{\alpha=1}^L\frac{H(S_\alpha)}{\alpha}-\sum_{\alpha=1}^L\frac{\delta_\alpha^{(n)}}{\alpha}.
\end{equation}
Finally, letting $n \rightarrow \infty$ and $\epsilon \rightarrow 0$ completes the proof of \eqref{eq:SR-Conv}, i.e., superposition coding can achieve the minimum sum rate for the general SMDC problem.

Note that unlike the original proof of \cite{RYH-IT97}, which uses the classical subset entropy inequality of Han \cite{Han-IC78} and hence involves \emph{all} nonempty subsets $U$ of $\Omega_L$, our proof relies on the sliding-window subset entropy inequality \eqref{eq:SW} and hence only involves the subsets $U$ of a sliding-window type, i.e., $U=W_l^{(\alpha)}$ for some $l=1,\ldots,L$ and $\alpha=1,\ldots,L$. Therefore, based on our proof, the converse result \eqref{eq:SR-Conv} remains to be true even if we weaken the asymptotically perfect reconstruction requirement \eqref{eq:PR} to 
\begin{equation}
\mathrm{Pr}\left\{d_U(X_U) \neq (S_1^n,\ldots,S_{|U|}^n)\right\} \leq \epsilon, 
\qquad \forall U \in \left\{W_l^{(\alpha)}: l=1,\ldots,L\; \mbox{and} \; \alpha=1,\ldots,L\right\}.
\end{equation}
This is the \emph{definitive} advantage of our proof over that based on the classical subset entropy inequality of Han \cite{Han-IC78}.

\subsection{The Subset Entropy Inequality of Yeung and Zhang Revisited}\label{sec:YZR}
In this section, we revisit the subset entropy inequality of Yeung and Zhang \eqref{eq:YZ}, which played a key in their proof \cite{YZ-IT99} of the optimality of superposition coding for achieving the entire admissible rate region of the problem. As mentioned previously, in \cite{YZ-IT99} the subset entropy inequality \eqref{eq:YZ} was proved by combining the classical subset entropy inequality of Han \cite{Han-IC78} and a number of analysis results on the sequence of linear programs \eqref{eq:dual2}. However, the inequality, as stated in Theorem~\ref{thm:YZ}, does not even directly imply the classical subset entropy inequality of Han \cite{Han-IC78}. The reason is that Theorem~\ref{thm:YZ} merely asserts the existence of a set of optimal solutions $c_{\boldsymbol{\lambda}}^{(\alpha)}$, $\alpha=1,\ldots,L$, that satisfies the subset entropy inequality \eqref{eq:YZ}, rather than providing a sufficient condition for the inequality to hold. Below, we shall use a subset entropy inequality recently proved by Madiman and Tetali \cite{MT-IT10} to summarize the analysis results of \cite{YZ-IT99} on the sequence of linear programs \eqref{eq:dual2} into a succinct sufficient condition for the subset entropy inequality \eqref{eq:YZ} to hold. 

\subsubsection{A Subset Entropy Inequality of Madiman and Tetali}
Consider a \emph{hypergraph} $(U,\mathcal{V})$ where $U$ is a finite ground set and $\mathcal{V}$ is a collection of subsets of $U$. A function $g: \mathcal{V} \rightarrow \mathbb{R}^+$ is called a \emph{fractional cover} of $(U,\mathcal{V})$ if it satisfies 
\begin{equation}
\sum_{\{V \in \mathcal{V}: V \ni i\}}g(V) \geq 1, \quad \forall i \in U.
\end{equation}

\begin{theorem}[A subset entropy inequality of Madiman and Tetali \cite{MT-IT10}]\label{thm:MT}
Let $(U,\mathcal{V})$ be a hypergraph, and let $g$ be a fractional cover of $(U,\mathcal{V})$. Then
\begin{equation}
\sum_{V \in \mathcal{V}}g(V)H(X_V) \geq H(X_U) 
\label{eq:MT}
\end{equation}
for any collection of jointly distributed random variables $X_U$.
\end{theorem}

The following corollary provides a ``chain" form of the subset entropy inequality \eqref{eq:MT}. Let $M$ be a positive integer, and let $\Sigma$ be a finite ground set. Let $\Sigma^{(\alpha)}$ be a collection of subsets of $\Sigma$ for each $\alpha=1,\ldots,M,$. Assuming that $\Sigma^{(\alpha)}$, $\alpha=1,\ldots,M$, are \emph{mutually exclusive}, then $\{\Sigma^{(\alpha)}:\alpha=1,\ldots,M\}$ \emph{induces} a collection of hypergraphs $\{(U,\mathcal{V}_U): U \in \cup_{\alpha=2}^M\Sigma^{(\alpha)}\}$ where
\begin{align}
\mathcal{V}_U := \{V\in \Sigma^{(\alpha-1)}:V \subseteq U\}, \quad \forall U \in \Sigma^{(\alpha)}.
\end{align}
We shall term each subset $V \in \mathcal{V}_U$ a ``child" of $U$. For convenience, we shall also define
\begin{align}
\mathcal{U}_V := \{U\in \Sigma^{(\alpha)}:U \supseteq V\}, \quad \forall V \in \Sigma^{(\alpha-1)}
\end{align}
and term each subset $U \in \mathcal{U}_V$ a ``parent" of $V$. 

\begin{coro}\label{cor:MT}
Let $c: \cup_{\alpha=1}^M\Sigma^{(\alpha)} \rightarrow \mathbb{R}^+$. For any $\alpha=2,\ldots,M$, if there exists a collection of functions $\{g_U: U \in \Sigma^{(\alpha)}\}$ for which each $g_U$ is a fractional cover of $(U,\mathcal{V}_U)$ and such that
\begin{equation}
c(V) = \sum_{U \in \mathcal{U}_V}g_U(V)c(U), \quad \forall V \in \Sigma^{(\alpha-1)}
\label{eq:MT2}
\end{equation}
we have
\begin{equation}
\sum_{V \in \Sigma^{(\alpha-1)}}c(V)H(X_V) \geq \sum_{U \in \Sigma^{(\alpha)}}c(U)H(X_U)
\label{eq:MT3}
\end{equation}
for any collection of jointly distributed random variables $X_\Sigma$.
\end{coro}

\begin{proof}
Fix $\alpha \in \{2,\ldots,M\}$. For any $U \in \Sigma^{(\alpha)}$, $g_U$ is a fractional cover of $(U,\mathcal{V}_U)$. By the subset entropy inequality of Madiman and Tetali \eqref{eq:MT}, we have
\begin{align}
\sum_{V \in \mathcal{V}_U}g_U(V)H(X_V) \geq H(X_U), \quad \forall U \in \Sigma^{(\alpha)}.
\label{eq:MT4}
\end{align}
Multiplying both sides of \eqref{eq:MT4} by $c(U)$ and summing over $U \in \Sigma^{(\alpha)}$, we have
\begin{align}
\sum_{U \in \Sigma^{(\alpha)}}\sum_{V \in \mathcal{V}_U}c(U)g_U(V)H(X_V) \geq \sum_{U \in \Sigma^{(\alpha)}}c(U)H(X_U).
\label{eq:MT5}
\end{align}
Note that
\begin{align}
\sum_{U \in \Sigma^{(\alpha)}}\sum_{V \in \mathcal{V}_U}c(U)g_U(V)H(X_V) &=\sum_{V \in \Sigma^{(\alpha-1)}}\left(\sum_{U \in \mathcal{U}_V}g_U(V)c(U)\right)H(X_V)\\
&=\sum_{V \in \Sigma^{(\alpha-1)}}c(V)H(X_V)\label{eq:MT6}
\end{align}
where \eqref{eq:MT6} follows \eqref{eq:MT2}. Substituting \eqref{eq:MT6} into \eqref{eq:MT5} completes the proof of the corollary.
\end{proof}

\subsubsection{Connections to the Subset Entropy Inequalities of Han and Yeung--Zhang}
Specifying $\Sigma=\Omega_L$, $M=L$, and $\Sigma^{(\alpha)}=\Omega_L^{(\alpha)}$ for $\alpha=1,\ldots,L$, the subset entropy inequality of Madiman and Tetali can be used to provide a \emph{unifying} proof for both the subset entropy inequality of Han and the subset entropy inequality of Yeung and Zhang. Note that the choice $\{\Sigma^{(\alpha)}=\Omega_L^{(\alpha)}:\alpha=1,\ldots,L\}$ is \emph{regular} in that each subset $U \in \Omega_L^{(\alpha)}$ has exactly $\alpha$ children in $\Omega_L^{(\alpha-1)}$, and each subset $V \in \Omega_L^{(\alpha-1)}$ has exactly $L-(\alpha-1)$ parents in $\Omega_L^{(\alpha)}$.

To see how the subset entropy inequality of Madiman and Tetali \eqref{eq:MT} implies the subset entropy inequality of Han \eqref{eq:Han}, let 
\begin{align}
c(U):=\frac{1}{\alpha
\left(
\begin{array}{c}
  L   \\
  \alpha   
\end{array}
\right)
}, \quad \forall U \in \Omega_L^{(\alpha)}\; \mbox{and} \; \alpha=1,\ldots,L\label{eq:cHan}
\end{align}
and
\begin{align}
g_U(V):=\frac{1}{\alpha-1}, \quad \forall U \in \Omega_L^{(\alpha)}, \; V \in \mathcal{V}_U, \; \mbox{and} \; \alpha=2,\ldots,L.
\end{align}
For any $\alpha=2,\ldots,L$ and $U \in \Omega_L^{(\alpha)}$,
\begin{align}
\sum_{\{V \in \mathcal{V}_U: V \ni i\}}g_U(V) =\frac{|\{V \in \mathcal{V}_U: V \ni i\}|}{\alpha-1}=\frac{\alpha-1}{\alpha-1}=1, \quad \forall i \in U
\end{align}
so $g_U$ is a \emph{uniform} fractional cover of $(U,\mathcal{V}_U)$. Furthermore, for any $\alpha=2,\ldots,L$ and $V \in \Omega_L^{(\alpha-1)}$ we have
\begin{align}
\sum_{U \in \mathcal{U}_V}g_U(V)c(U)&=
\frac{|\mathcal{U}_V|}{(\alpha-1)\alpha
\left(
\begin{array}{c}
  L   \\
  \alpha   
\end{array}
\right)
}
=
\frac{L-(\alpha-1)}{(\alpha-1)\alpha
\left(
\begin{array}{c}
  L   \\
  \alpha   
\end{array}
\right)
}
=
\frac{1}{(\alpha-1)
\left(
\begin{array}{c}
  L   \\
  \alpha-1   
\end{array}
\right)
}=c(V).
\end{align}
Substituting \eqref{eq:cHan} into \eqref{eq:MT3} immediately gives the subset entropy inequality of Han \eqref{eq:Han}.

To see how the subset entropy inequality of Madiman and Tetali \eqref{eq:MT} implies the subset entropy inequality of Yeung and Zhang \eqref{eq:YZ}, we shall need the following result, which is a synthesis of the analytical results on the sequence of linear programs \eqref{eq:dual2} established in \cite{YZ-IT99}. (For completeness, a sketched proof based on the results of \cite{YZ-IT99} is included in Appendix~\ref{app:A}.) 

\begin{theorem}[A linear programing result of Yeung and Zhang \cite{YZ-IT99}]\label{thm:YZ2}
For any $\boldsymbol{\lambda} \in (\mathbb{R}^+)^L$, any $\alpha=2,\ldots,L$, and any $c_{\boldsymbol{\lambda}}^{(\alpha)}$ which is an optimal solution to the linear program \eqref{eq:dual2} with the optimal value $f_{\alpha}(\boldsymbol{\lambda})>0$, there exists a collection of functions $\{g_U: U \in \Omega_L^{(\alpha)}\}$ for which each $g_U$ is a fractional cover of $(U,\mathcal{V}_U)$ and such that $c_{\boldsymbol{\lambda}}^{(\alpha-1)}=\{c_{\boldsymbol{\lambda}}(V): V \in \Omega_L^{(\alpha-1)}\}$ where
\begin{equation}
c_{\boldsymbol{\lambda}}(V) := \sum_{U \in \mathcal{U}_V}g_U(V)c_{\boldsymbol{\lambda}}(U)
\label{eq:YZ2}
\end{equation}
is an optimal solution to the linear program \eqref{eq:dual2} with $\alpha$ replaced by $\alpha-1$.
\end{theorem}

Now fix $\boldsymbol{\lambda} \in (\mathbb{R}^+)^L$, and consider the following construction of $c_{\boldsymbol{\lambda}}=\cup_{\alpha=1}^Lc_{\boldsymbol{\lambda}}^{(\alpha)}$. For $\alpha=L$, choose $c_{\boldsymbol{\lambda}}^{(L)}$ to be an arbitrary \emph{optimal} solution to the linear program \eqref{eq:dual2}. For $\alpha=1,\ldots,L-1$, construct $c_{\boldsymbol{\lambda}}^{(\alpha)}$ iteratively as follows. Suppose that $c_{\boldsymbol{\lambda}}^{(\alpha)}$ is already in place for some  $\alpha=2,\ldots,L$ such that $c_{\boldsymbol{\lambda}}^{(\alpha)}$ is an optimal solution to the linear program \eqref{eq:dual2}. If the optimal value $f_{\alpha}(\boldsymbol{\lambda})>0$, construct $c_{\boldsymbol{\lambda}}^{(\alpha-1)}$ according to \eqref{eq:MT2} so $c_{\boldsymbol{\lambda}}^{(\alpha-1)}$ is an optimal solution to the linear program \eqref{eq:dual2} with $\alpha$ replaced by $\alpha-1$. Moreover, by Corollary~\ref{cor:MT} $c_{\boldsymbol{\lambda}}^{(\alpha-1)}$ and $c_{\boldsymbol{\lambda}}^{(\alpha)}$ satisfy the subset entropy inequality of Yeung and Zhang \eqref{eq:YZ}. If, on the other hand, $f_{\alpha}(\boldsymbol{\lambda})=0$, we have $c_{\boldsymbol{\lambda}}(U)=0$ for \emph{all} $U \in \Omega_L^{(\alpha)}$. In this case, choose $c_{\boldsymbol{\lambda}}^{(\alpha-1)}$ to be an arbitrary \emph{optimal} solution to the linear program \eqref{eq:dual2} with $\alpha$ replaced by $\alpha-1$, and $c_{\boldsymbol{\lambda}}^{(\alpha-1)}$ and $c_{\boldsymbol{\lambda}}^{(\alpha)}$ will trivially satisfy the subset entropy inequality of Yeung and Zhang \eqref{eq:YZ}. We have thus constructed for any $\boldsymbol{\lambda} \in (\mathbb{R}^+)^L$, a sequence of $c_{\boldsymbol{\lambda}}^{(\alpha)}$, $\alpha=1,\ldots,L$, such that each $c_{\boldsymbol{\lambda}}^{(\alpha)}$ is an optimal solution to the linear program \eqref{eq:dual2}, and the subset entropy inequality of Yeung and Zhang \eqref{eq:YZ} holds for each $\alpha=2,\ldots,L$.

We mention here that even though both the subset entropy inequality of Han and the subset entropy inequality of Yeung and Zhang can be directly established from the subset entropy inequality of Madiman and Tetali, this is \emph{not} the case for the sliding-window subset entropy inequality \eqref{eq:SW} except for $\alpha=2$ and $L$. This can be seen as follows. 

Let $\Sigma=\Omega_L$, $M=L$, and $\Sigma^{(\alpha)}=\{W_l^{(\alpha)}:l=1,\ldots,L\}$ for $\alpha=1,\ldots,L$. Note that for any $\alpha=1,\ldots,L-1$, each sliding window $W_l^{(\alpha)}$ represents a \emph{different} subset for different $l$. (For $\alpha=L$, all sliding windows $W_l^{(L)}$, $l=1,\ldots,L$, represent the \emph{same} subset $\Omega_L$.) Furthermore, for any $\alpha=2,\ldots,L-1$ each sliding window $W_l^{(\alpha)}$ has only two children: $W_l^{(\alpha-1)}$ and $W_{\langle l+1\rangle}^{(\alpha-1)}$, and each sliding window $W_l^{(\alpha-1)}$ has only two parents: $W_l^{(\alpha)}$ and $W_{\langle l-1\rangle}^{(\alpha)}$. Now consider the elements $l$ and $\langle l+\alpha-1 \rangle$ from $W_l^{(\alpha)}$. Note that among the two children $W_l^{(\alpha-1)}$ and $W_{\langle l+1\rangle}^{(\alpha-1)}$ of $W_l^{(\alpha)}$, $l$ belongs only to $W_l^{(\alpha-1)}$, and $\langle l+\alpha-1 \rangle$ belong only to $W_{\langle l+1\rangle}^{(\alpha-1)}$. Thus, \emph{any} fractional cover $g_{W_l^{(\alpha)}}$ of the hypergraph $(W_l^{(\alpha)},\{W_l^{(\alpha-1)},W_{\langle l+1\rangle}^{(\alpha-1)}\})$ must satisfy 
\begin{align}
g_{W_l^{(\alpha)}}(W_l^{(\alpha-1)}) \geq 1 \quad \mbox{and} \quad g_{W_l^{(\alpha)}}(W_{\langle l+1\rangle}^{(\alpha-1)}) \geq 1.
\end{align}
Now let $c(W_l^{(\alpha)}):=1/\alpha$ for all $l=1,\ldots,L$ and $\alpha=1,\ldots,L-1$. We have
\begin{align}
g_{W_l^{(\alpha)}}(W_l^{(\alpha-1)})c(W_l^{(\alpha)})+g_{W_{\langle l-1\rangle}^{(\alpha)}}(W_l^{(\alpha-1)})c(W_{\langle l-1\rangle}^{(\alpha)}) \geq \frac{2}{\alpha} > \frac{1}{\alpha-1} = c(W_l^{(\alpha-1)})
\end{align}
for any $\alpha>2$. We thus conclude that for any $2<\alpha<L$, the sliding-window subset entropy inequality \eqref{eq:SW} cannot be directly inferred from the subset entropy inequality of Madiman and Tetali.

\subsubsection{A Conditional Subset Entropy Inequality of Yeung and Zhang}
We conclude this section by providing a conditional extension of the subset entropy inequality of Yeung and Zhang, which will play a key role in proving the optimality of superposition coding for achieving the entire admissible rate region of the general S-SMDC problem. We shall start with the following generalization of Corollary~\ref{cor:MT}.

Let $\Sigma$ be a finite ground set, and let $\Sigma^{(\alpha)}$, $\alpha=1,\ldots,M$, be a collection of subsets of $\Sigma$. As before, we shall assume that the collections $\Sigma^{(\alpha)}$, $\alpha=1,\ldots,M$, are mutually exclusive, so $\{\Sigma^{(\alpha)}: \alpha=1,\ldots,M\}$ induces a hypergraph $(U,\mathcal{V}_U)$ for every $U \in \cup_{\alpha=2}^{M}\Sigma^{(\alpha)}$. For each $U \in \Sigma^{(M)}$ let $\mathcal{A}_U$ be a collection of subsets of $\Sigma$, and let $\mathcal{A}^{(M)}:=\{\mathcal{A}_U:U \in \Sigma^{(M)}\}$. For $\alpha=1,\ldots,M-1$, define $\mathcal{A}^{(\alpha)}:=\{\mathcal{A}_U: U \in \Sigma^{(\alpha)}\}$ iteratively as follows. Suppose that $\mathcal{A}^{(\alpha)}$ is already in place for some $\alpha=2,\ldots,M$. Let $\mathcal{A}^{(\alpha-1)}=\{\mathcal{A}_V: V \in \Sigma^{(\alpha-1)}\}$ where
\begin{align}
\mathcal{A}_V := \cup_{U \in \mathcal{U}_V}\mathcal{A}_U.
\label{eq:Av}
\end{align}

\begin{prop}\label{prop:cMT}
For each $U \in \cup_{\alpha=1}^M\Sigma^{(\alpha)}$, let $s(U,\cdot): \mathcal{A}_U \rightarrow \mathbb{R}^+$. For any $\alpha=2,\ldots,M$, if there exists a collection of functions $\{g_U: U \in \Sigma^{(\alpha)}\}$ for which each $g_U$ is a fractional cover of $(U,\mathcal{V}_U)$ and such that
\begin{equation}
s(V,A) = \sum_{\{U \in \mathcal{U}_V: \mathcal{A}_U \ni A\}}g_U(V)s(U,A), \quad \forall V \in \Sigma^{(\alpha-1)}\; \mbox{and} \; A \in \mathcal{A}_V
\label{eq:cMT}
\end{equation}
we have
\begin{equation}
\sum_{V \in \Sigma^{(\alpha-1)}}\sum_{A \in \mathcal{A}_V}s(V,A)H(X_V|X_A) \geq 
\sum_{U \in \Sigma^{(\alpha)}}\sum_{A \in \mathcal{A}_U}s(U,A)H(X_U|X_A)
\label{eq:cMT2}
\end{equation}
for any collection of jointly distributed random variables $X_\Sigma$.
\end{prop}

\begin{proof}
Fix $\alpha \in \{2,\ldots,M\}$. For any $U \in \Sigma^{(\alpha)}$, $g_U$ is a fractional cover of $(U,\mathcal{V}_U)$. By the subset entropy inequality of Madiman and Tetali \eqref{eq:MT}, we have
\begin{align}
\sum_{V \in \mathcal{V}_U}g_U(V)H(X_V|X_A) \geq H(X_U|X_A), \quad \forall U \in \Sigma^{(\alpha)} \; \mbox{and} \; A \in \mathcal{A}_U.
\label{eq:cMT3}
\end{align}
Multiplying both sides of \eqref{eq:cMT3} by $s(U,A)$ and summing over $A \in \mathcal{A}_U$ and $U \in \Sigma^{(\alpha)}$, we have
\begin{align}
\sum_{U \in \Sigma^{(\alpha)}}\sum_{A \in \mathcal{A}_U}\sum_{V \in \mathcal{V}_U}s(U,A)g_U(V)H(X_V|X_A) \geq \sum_{U \in \Sigma^{(\alpha)}}\sum_{A \in \mathcal{A}_U}s(U,A)H(X_U|X_A).
\label{eq:cMT4}
\end{align}
Note that
\begin{align}
\sum_{U \in \Sigma^{(\alpha)}}&\sum_{A \in \mathcal{A}_U}\sum_{V \in \mathcal{V}_U}s(U,A)g_U(V)H(X_V|X_A)\notag\\
&=\sum_{U \in \Sigma^{(\alpha)}}\sum_{V \in \mathcal{V}_U}\sum_{A \in \mathcal{A}_U}s(U,A)g_U(V)H(X_V|X_A)\\
&=\sum_{V \in \Sigma^{(\alpha-1)}}\sum_{U \in \mathcal{U}_V}\sum_{A \in \mathcal{A}_U}s(U,A)g_U(V)H(X_V|X_A)\\
&=\sum_{V \in \Sigma^{(\alpha-1)}}\sum_{A \in \mathcal{A}_V}\left(\sum_{\{U \in \mathcal{U}_V: \mathcal{A}_U \ni A\}}s(U,A)g_U(V)\right)H(X_V|X_A)\\
&=\sum_{V \in \Sigma^{(\alpha-1)}}\sum_{A \in \mathcal{A}_V}s(V,A)H(X_V|X_A)\label{eq:cMT5}
\end{align}
where \eqref{eq:cMT5} follows from \eqref{eq:cMT}. Substituting \eqref{eq:cMT5} into \eqref{eq:cMT4} completes the proof of the proposition.
\end{proof}

\begin{theorem}[A conditional subset entropy inequality of Yeung and Zhang]\label{thm:CYZ}
For any $\boldsymbol{\lambda} \in (\mathbb{R}^+)^L$ and $N=0,\ldots,L-1$, there exists for each $U \in \cup_{\alpha=1}^{L-N} \Omega_L^{(\alpha)}$ a collection of subsets $\mathcal{A}_U$ of $\Omega_L$ such that: 
\begin{align}
|A|=N \quad \mbox{and} \quad A \cap U =\emptyset, \quad \forall A \in \mathcal{A}_U
\end{align} 
and a function $s_{\boldsymbol{\lambda}}(U,\cdot): \mathcal{A}_U \rightarrow \mathbb{R}^+$ such that:  
\begin{itemize}
\item[1)] for each $\alpha=1,\ldots,L-N$, $c_{\boldsymbol{\lambda}}^{(\alpha)}=\{c_{\boldsymbol{\lambda}}(U):U \in \Omega_L^{(\alpha)}\}$ where
\begin{align}
c_{\boldsymbol{\lambda}}(U) := \sum_{A \in \mathcal{A}_U}s_{\boldsymbol{\lambda}}(U,A)
\label{eq:CYZ1}
\end{align}
is an \emph{optimal} solution to the linear program \eqref{eq:dual2}; and
\item[2)] for each $\alpha=2,\ldots,L-N$,
\begin{align}
\sum_{V \in \Omega_L^{(\alpha-1)}}\sum_{A \in \mathcal{A}_V}s(V,A)H(X_V|X_A) \geq 
\sum_{U \in \Omega_L^{(\alpha)}}\sum_{A \in \mathcal{A}_U}s(U,A)H(X_U|X_A)\label{eq:CYZ2}
\end{align}
for \emph{any} collection of $L$ jointly distributed random variables $(X_1,\ldots,X_L)$.
\end{itemize}
\end{theorem}

\begin{proof}
Fix $\boldsymbol{\lambda} \in (\mathbb{R}^+)^L$ and $N \in \{0,\ldots,L-1\}$, and let $\Sigma=\Omega_L$, $M=L-N$, and $\Sigma^{(\alpha)}=\Omega_L^{(\alpha)}$ for $\alpha=1,\ldots,L-N$. Consider the following construction of $\mathcal{A}^{(\alpha)}$ and $s_{\boldsymbol{\lambda}}^{(\alpha)}:=\{s(U,\cdot):U\in\Omega_L^{(\alpha)}\}$, $\alpha=1,\ldots,L-N$.

For $\alpha=L-N$, let $\mathcal{A}^{(L-N)}=\{\mathcal{A}_U:U \in \Omega_L^{(L-N)}\}$ where $\mathcal{A}_U:=\{\Omega_L\setminus U\}$, i.e., each $\mathcal{A}_U$ contains a single subset $A=\Omega_L\setminus U$ of size $|A|=L-(L-N)=N$ and such that $A\cap U=\emptyset$. Furthermore, let $c_{\boldsymbol{\lambda}}^{(L-N)}=\{c_{\boldsymbol{\lambda}}(U): U \in \Omega_L^{(L-N)}\}$ be an \emph{optimal} solution to the linear program \eqref{eq:dual2} for $\alpha=L-N$, and let
\begin{align}
s_{\boldsymbol{\lambda}}(U,\Omega_L\setminus U) := c_{\boldsymbol{\lambda}}(U), \quad \forall U \in \Omega_L^{(L-N)}.
\end{align}
Since by construction each $\mathcal{A}_U$, $U \in \Omega_L^{(L-N)}$, contains a \emph{single} subset $A=\Omega_L\setminus U$, we trivially have
\begin{align}
\sum_{A \in \mathcal{A}_U}s_{\boldsymbol{\lambda}}(U,A)=c_{\boldsymbol{\lambda}}(U), \quad \forall U \in \Omega_L^{(L-N)}.
\end{align}

For $\alpha=1,\ldots,L-N-1$, let us construct $\mathcal{A}^{(\alpha)}$ and $s_{\boldsymbol{\lambda}}^{(\alpha)}$ iteratively as follows. Suppose that $\mathcal{A}^{(\alpha)}$ and $s_{\boldsymbol{\lambda}}^{(\alpha)}$ are already in place for some $\alpha=2,\ldots,L-N$ such that $|A|=N$ and $A\cap U =\emptyset$ for any $U \in \Omega_L^{(\alpha)}$ and $A \in \mathcal{A}_U$, and $c_{\boldsymbol{\lambda}}^{(\alpha)}=\{c_{\boldsymbol{\lambda}}(U):U \in \Omega_L^{(\alpha)}\}$ where $c_{\boldsymbol{\lambda}}(U)$ is given by \eqref{eq:CYZ1} is an \emph{optimal} solution to the linear program \eqref{eq:dual2}. First, construct $\mathcal{A}^{(\alpha-1)}$ according to \eqref{eq:Av}. Based on this construction, for any $V \in \Omega_L^{(\alpha-1)}$ and $A \in \mathcal{A}_V$ we have $A \in \mathcal{A}_U$ for some $U \in \mathcal{U}_V \subseteq \Omega_L^{(\alpha)}$. Therefore, by the induction assumption we must have $|A|=N$ and
\begin{align}
A\cap V \subseteq A\cap U = \emptyset
\end{align}
for any $V \in \Omega_L^{(\alpha-1)}$ and $A \in \mathcal{A}_V$. Next, construct $s_{\boldsymbol{\lambda}}^{(\alpha-1)}$ as follows. If the optimal value $f_\alpha(\boldsymbol{\lambda})>0$, by Theorem~\ref{thm:YZ2} there exists a collection of functions $\{g_U: U \in \Omega_L^{(\alpha)}\}$ for which each $g_U$ is a fractional cover of $(U,\mathcal{V}_U)$ and such that $c_{\boldsymbol{\lambda}}^{(\alpha-1)}=\{c_{\boldsymbol{\lambda}}(V): V \in \Omega_L^{(\alpha-1)}\}$ where
$c_{\boldsymbol{\lambda}}(V)$ is given by \eqref{eq:YZ2} is an \emph{optimal} solution to the linear program \eqref{eq:dual2} with $\alpha$ replaced by $\alpha-1$. In this case, let $s_{\boldsymbol{\lambda}}^{(\alpha-1)}=\{s_{\boldsymbol{\lambda}}(V,\cdot): V \in \Omega_L^{(\alpha-1)}\}$ where 
\begin{align}
s_{\boldsymbol{\lambda}}(V,A) := \sum_{\{U \in \mathcal{U}_V: \mathcal{A}_U \ni A\}}g_U(V)s_{\boldsymbol{\lambda}}(U,A).
\end{align}
Thus, for each $V \in \Omega_L^{(\alpha-1)}$ we have
\begin{align}
\sum_{A \in \mathcal{A}_V}s_{\boldsymbol{\lambda}}(V,A)&=\sum_{A \in \mathcal{A}_V}\left[\sum_{\{U \in \mathcal{U}_V: \mathcal{A}_U \ni A\}}g_U(V)s_{\boldsymbol{\lambda}}(U,A)\right]\\
&= \sum_{U \in \mathcal{U}_V}g_U(V)\left[\sum_{A \in \mathcal{A}_U}s_{\boldsymbol{\lambda}}(U,A)\right]\\
&= \sum_{U \in \mathcal{U}_V}g_U(V)c_{\boldsymbol{\lambda}}(U)\\
&=c_{\boldsymbol{\lambda}}(V)
\end{align}
Furthermore, by Proposition~\ref{prop:cMT} $s_{\boldsymbol{\lambda}}^{(\alpha-1)}$ and $s_{\boldsymbol{\lambda}}^{(\alpha-1)}$ satisfy the subset entropy inequality \eqref{eq:CYZ2}. If, on the other hand, $f_\alpha(\boldsymbol{\lambda})=0$, we have $s_{\boldsymbol{\lambda}}(U,A)=0$ for \emph{all} $U \in \Omega_L^{(\alpha)}$ and $A \in \mathcal{A}_U$. In this case, choose an arbitrary $s_{\boldsymbol{\lambda}}^{(\alpha-1)}$ such that $c_{\boldsymbol{\lambda}}^{(\alpha-1)}=\{c_{\boldsymbol{\lambda}}(V): V \in \Omega_L^{(\alpha-1)}\}$ where
\begin{align}
c_{\boldsymbol{\lambda}}(V) := \sum_{A \in \mathcal{A}_V}s_{\boldsymbol{\lambda}}(V,A)
\end{align}
is an optimal solution to the linear program \eqref{eq:dual2} with $\alpha$ being replaced by $\alpha-1$, and $s_{\boldsymbol{\lambda}}^{(\alpha-1)}$ and $s_{\boldsymbol{\lambda}}^{(\alpha)}$ will trivially satisfy the subset entropy inequality \eqref{eq:CYZ2}.

We have thus constructed for any $\boldsymbol{\lambda} \in (\mathbb{R}^+)^L$ and $N=0,\ldots,L-1$, a sequence of $\mathcal{A}^{(\alpha)}$ and $c_{\boldsymbol{\lambda}}^{(\alpha)}$, $\alpha=1,\ldots,L-N$, such that all conditions of Theorem~\ref{thm:CYZ} are met simultaneously. This completes the proof of the theorem.
\end{proof}

\section{Two Extensions: SMDC-A and S-SMDC}\label{sec:Ext}
\subsection{Extension 1: SMDC-A}\label{sec:Ext-SMDCA}
\subsubsection{Problem Statement}\label{sec:Ext-SMDCA-PS}

\begin{figure}[!t]
\centering
\includegraphics[width=0.95\linewidth,draft=false]{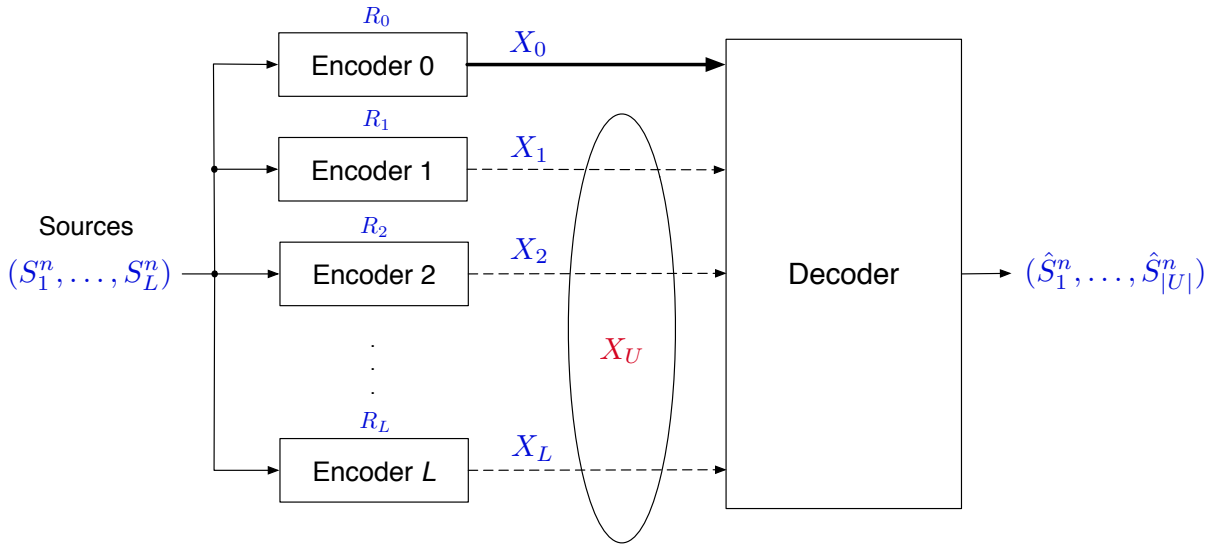}
\caption{SMDC with an all-access encoder 0 and $L$ randomly accessible encoders $1$ to $L$. A total of $L$ independent discrete memoryless sources $(S_1,\ldots,S_L)$ are to be encoded at the encoders. The decoder, which has access to encoder 0 and a subset $U$ of the randomly accessible encoders, needs to nearly perfectly reconstruct the sources $(S_1,\ldots,S_{|U|})$ no matter what the realization of $U$ is.}
\label{fig:model}
\end{figure}

As illustrated in Figure~\ref{fig:model}, the problem of SMDC-A consists of:
\begin{itemize}
\item a total of $L$ \emph{independent} discrete memoryless sources $\{S_\alpha[t]\}_{t=1}^{\infty}$, where $\alpha=1,\ldots,L$ and $t$ is the time index;
\item a set of $L+1$ encoders (encoder $0$ to $L$);
\item a decoder who has access to a subset $\{0\}\cup U$ of the encoder outputs for some nonempty $U \subseteq \Omega_L$.
\end{itemize}
The realization of $U$ is \emph{unknown} a priori at the encoders. However, no matter which $U$ actually materializes, the decoder needs to nearly perfectly reconstruct the sources $(S_1,\ldots,S_\alpha)$ whenever $|U| \geq \alpha$. 

Formally, an $(n,(M_0,M_1,\ldots,M_L))$ code is defined by a collection of $L+1$ encoding functions
\begin{equation}
e_l: \prod_{\alpha=1}^{L}\mathcal{S}_\alpha^n \rightarrow \{1,\ldots,M_l\}, \quad \forall l=0,1,\ldots,L
\end{equation}
and $2^L-1$ decoding functions
\begin{equation}
d_U: \{1,\ldots,M_0\}\times \prod_{l\in U}\{1,\ldots,M_l\} \rightarrow \prod_{\alpha=1}^{|U|}\mathcal{S}_\alpha^n, \quad \forall U \subseteq \Omega_L\; \mbox{s.t.}\; U \neq \emptyset.
\end{equation}
A nonnegative rate tuple $(R_0,R_1,\ldots,R_L)$ is said to be \emph{admissible} if for every $\epsilon>0$, there exits, for sufficiently large block length $n$, an $(n,(M_0,M_1,\ldots,M_L))$ code such that:
\begin{itemize}
\item (Rate constraints at the encoders)
\begin{equation}
\frac{1}{n}\log M_l \leq R_l +\epsilon, \qquad \forall l =0,1,\ldots,L;\label{eq:Rate-A}
\end{equation}
\item (Asymptotically perfect reconstructions at the decoder)
\begin{equation}
\mathrm{Pr}\left\{d_U(X_{\{0\}\cup U}) \neq (S_1^n,\ldots,S_{|U|}^n)\right\} \leq \epsilon, 
\qquad \forall U \subseteq \Omega_L\; \mbox{s.t.}\; U \neq \emptyset
\label{eq:PR-A}
\end{equation}
where $S_\alpha^n := \{S_\alpha[t]\}_{t=1}^{n}$, $X_l:=e_l(S_1^n,\ldots,S_L^n)$ is the output of encoder $l$, and $X_{\{0\}\cup U}:=\{X_l: l \in \{0\}\cup U\}$.
\end{itemize}
The \emph{admissible rate region} $\mathcal{R}$ is the collection of \emph{all} admissible rate tuples $(R_0,R_1,\ldots,R_L)$. 

\subsubsection{Superposition Coding Rate Region}\label{sec:Ext-SMDCA-SC}
Similar to classical SMDC, a natural strategy for SMDC-A is superposition coding, i.e., to encode the sources separately at the encoders and there is no coding across different sources. Formally, the problem of encoding a single source $S_\alpha$ can be viewed as a special case of the general problem where the sources $S_m$ are deterministic for all $m \neq \alpha$. In this case, the source $S_\alpha$ needs to be nearly perfectly reconstructed whenever the decoder can access at least $\alpha$ randomly accessible encoders in addition to the all-access encoder $0$. The following scheme is a natural extension of the simple source-channel separation scheme considered previously for classical SMDC:

\begin{itemize}
\item First compress the source sequence $S_\alpha^n$ into a source message $W_\alpha$ using a \emph{lossless} source code. It is well known \cite[Ch.~5]{CT-B06} that the rate of the source message $W_\alpha$ can be made arbitrarily close to the entropy rate $H(S_\alpha)$ for sufficiently large block length $n$.
\item Next, divide the source message $W_\alpha$ into two independent sub-messages $W_\alpha^{(0)}$ and $W_\alpha^{(1)}$ so we have
\begin{align}
H(W_\alpha) &= H(W_\alpha^{(0)})+H(W_\alpha^{(1)}).
\label{eq:SH}
\end{align} 
The sub-message $W_\alpha^{(0)}$ is stored at the all-access encoder 0 \emph{without} any coding, which requires
\begin{align}
R_0 & \geq \frac{1}{n}H(W_\alpha^{(0)}).\label{eq:SH2}
\end{align} 
The sub-message $W_\alpha^{(1)}$ is encoded by the randomly accessible encoders 1 to $L$ using a \emph{maximum distance separable} code \cite{Sin-IT64}. Clearly, the sub-message $W_\alpha^{(1)}$ can be perfectly recovered at the decoder whenever 
\begin{equation}
\sum_{l \in U} R_l \geq \frac{1}{n}H(W_\alpha^{(1)}), \quad \forall U \in \Omega_L^{(\alpha)}\label{eq:SH3}
\end{equation}
for sufficiently large block length $n$. Eliminating $H(W_\alpha^{(0)})$ and $H(W_\alpha^{(1)})$ from \eqref{eq:SH}--\eqref{eq:SH3}, we conclude that the source message $W_\alpha$ can be perfectly recovered at the decoder whenever 
\begin{equation}
R_0+\sum_{l \in U} R_l \geq \frac{1}{n}H(W_\alpha), \quad \forall U \in \Omega_L^{(\alpha)}\label{eq:SH4}
\end{equation}
\end{itemize}

Combining the above two steps, we conclude that the rate region that can be achieved by the above source-channel separation scheme is given by the collection of all nonnegative rate tuples $(R_0,R_1,\ldots,R_L)$ satisfying
\begin{align}
R_0+\sum_{l \in U} R_l & \geq H(S_\alpha), \quad \forall U \in \Omega_L^{(\alpha)}.
\label{eq:SSDC-A}
\end{align}
Following the same footsteps as those for classical SMDC \cite{Roc-Thesis92,Yeu-IT95}, it is straightforward to show that the above rate region is in fact the admissible rate region for encoding the single source $S_\alpha$. By definition, the superposition coding rate region $\mathcal{R}_{sup}$ for SMDC-A is given by the collection of all nonnegative rate tuples $(R_0,R_1,\ldots,R_L)$ such that
\begin{equation}
R_l := \sum_{\alpha=1}^{L}r_l^{(\alpha)}
\end{equation}
for some nonnegative $r_l^{(\alpha)}$, $\alpha=1,\ldots,L$ and $l=0,1,\ldots,L$, satisfying  
\begin{equation}
r_0^{(\alpha)}+\sum_{l \in U} r_l^{(\alpha)} \geq H(S_\alpha), \quad \forall U \in \Omega_L^{(\alpha)}.
\end{equation}

Similar to classical SMDC, the superposition coding rate region $\mathcal{R}_{sup}$ for SMDC-A is a \emph{polyhedron} with polyhedral cone being the nonnegative orthant in $\mathbb{R}^{L+1}$ and hence can be completely characterized by the supporting hyperplanes 
\begin{equation}
\sum_{l=0}^L\lambda_lR_l \geq f(\lambda_0,\boldsymbol{\lambda}), \quad \forall \lambda_0 \geq 0 \; \mbox{and} \;\boldsymbol{\lambda}:=(\lambda_1,\ldots,\lambda_L) \in (\mathbb{R}^+)^L
\label{eq:sup-A}
\end{equation}
where
\begin{eqnarray}
f(\lambda_0,\boldsymbol{\lambda}) & = & \min_{(R_0,R_1,\ldots,R_L) \in \mathcal{R}_{sup}}\sum_{l=0}^{L}\lambda_lR_l\\
& = & 
\begin{array}{rl}
\min & \sum_{l=0}^L\left(\sum_{\alpha=1}^L\lambda_lr_l^{(\alpha)}\right)\\
\mbox{subject to} & r_0^{(\alpha)}+\sum_{l \in U} r_l^{(\alpha)} \geq H(S_\alpha), \quad \forall U \in \Omega_L^{(\alpha)} \; \mbox{and} \; \alpha=1,\ldots,L\\
& r_l^{(\alpha)} \geq 0, \quad \forall \alpha=1,\ldots,L \; \mbox{and} \; l=0,\ldots,L.
\end{array}
\end{eqnarray}
Clearly, the above optimization problem can be separated into the following $L$ sub-optimization problems:
\begin{equation}
f(\lambda_0,\boldsymbol{\lambda})=\sum_{\alpha=1}^Lf'_\alpha(\lambda_0,\boldsymbol{\lambda})
\end{equation}
where
\begin{eqnarray}
f'_\alpha(\lambda_0,\boldsymbol{\lambda}) & = & 
\begin{array}{rl}
\min & \sum_{l=0}^L\lambda_lr_l^{(\alpha)}\\
\mbox{subject to} & r_0^{(\alpha)}+\sum_{l \in U} r_l^{(\alpha)} \geq H(S_\alpha), \quad \forall U \in \Omega_L^{(\alpha)}\\
& r_l^{(\alpha)} \geq 0, \quad \forall l=0,\ldots,L
\end{array}\\
& = & 
\begin{array}{rl}
\max & \left(\sum_{U \in \Omega_L^{(\alpha)}}c_{\lambda_0,\boldsymbol{\lambda}}(U)\right)H(S_\alpha)\\
\mbox{subject to} &  \sum_{U \in \Omega_L^{(\alpha)}} c_{\lambda_0,\boldsymbol{\lambda}}(U) \leq \lambda_0\\
& \sum_{\{U \in \Omega_L^{(\alpha)}: U \ni l\}} c_{\lambda_0,\boldsymbol{\lambda}}(U) \leq \lambda_l, \quad \forall l=1,\ldots,L\\
& c_{\lambda_0,\boldsymbol{\lambda}}(U) \geq 0, \quad \forall U \in \Omega_L^{(\alpha)}.
\end{array}
\label{eq:dual-A}
\end{eqnarray}
Here, \eqref{eq:dual-A} follows from the strong duality for linear programs. For any $\lambda_0 \geq 0$, $\boldsymbol{\lambda} \in (\mathbb{R}^+)^L$, and $\alpha=1,\ldots,L$, let
\begin{eqnarray}
f_\alpha(\lambda_0,\boldsymbol{\lambda}) & := & 
\begin{array}{rl}
\max & \sum_{U \in \Omega_L^{(\alpha)}}c_{\lambda_0,\boldsymbol{\lambda}}(U)\\
\mbox{subject to} &  \sum_{U \in \Omega_L^{(\alpha)}} c_{\lambda_0,\boldsymbol{\lambda}}(U) \leq \lambda_0\\
& \sum_{\{U \in \Omega_L^{(\alpha)}: U \ni l\}} c_{\lambda_0,\boldsymbol{\lambda}}(U) \leq \lambda_l, \quad \forall l=1,\ldots,L\\
& c_{\lambda_0,\boldsymbol{\lambda}}(U) \geq 0, \quad \forall U \in \Omega_L^{(\alpha)}.
\end{array}
\label{eq:dual-A2}
\end{eqnarray}
We have $f'_\alpha(\lambda_0,\boldsymbol{\lambda}) =f_\alpha(\lambda_0,\boldsymbol{\lambda}) H(S_\alpha)$ and hence
\begin{equation}
f(\lambda_0,\boldsymbol{\lambda}) = \sum_{\alpha=1}^{L}f_\alpha(\lambda_0,\boldsymbol{\lambda})H(S_\alpha)
\label{eq:sup-A2}
\end{equation} 
for any $\lambda_0 \geq 0$ and $\boldsymbol{\lambda} \in (\mathbb{R}^+)^L$.

Note that in the optimization problem \eqref{eq:dual-A2}, if the constraint $\sum_{U \in \Omega_L^{(\alpha)}} c_{\lambda_0,\boldsymbol{\lambda}}(U) \leq \lambda_0$ is inactive, it can be removed from the program. In this case the optimal value $f_\alpha(\lambda_0,\boldsymbol{\lambda})=f_\alpha(\boldsymbol{\lambda})$, where $f_\alpha(\boldsymbol{\lambda})$ is the optimal value of the linear program \eqref{eq:dual2}. On the other hand, if the constraint $\sum_{U \in \Omega_L^{(\alpha)}} c_{\lambda_0,\boldsymbol{\lambda}} (U) \leq \lambda_0$ is active, the optimal value $f_\alpha(\lambda_0,\boldsymbol{\lambda})=\lambda_0$. Combing these two cases, we have 
\begin{equation}
f_\alpha(\lambda_0,\boldsymbol{\lambda})=\min(f_\alpha(\boldsymbol{\lambda}),\lambda_0), \quad \forall \alpha=1,\ldots,L.
\label{eq:sup-A2.5}
\end{equation}
Substituting \eqref{eq:sup-A2} and \eqref{eq:sup-A2.5} into \eqref{eq:sup-A}, we conclude that the superposition coding rate region $\mathcal{R}_{sup}$ for SMDC with an all-access encoder is given by the collection of all nonnegative rate tuples $(R_0,R_1,\ldots,R_L)$ satisfying
\begin{equation}
\sum_{l=0}^L\lambda_lR_l \geq \sum_{\alpha=1}^L \min(f_\alpha(\boldsymbol{\lambda}),\lambda_0)H(S_\alpha), \quad \forall \lambda_0 \geq 0\; \mbox{and} \; \boldsymbol{\lambda} \in (\mathbb{R}^+)^L.
\label{eq:sup-A3}
\end{equation}

As mentioned previously, the superposition coding rate region $\mathcal{R}_{sup}$ is a polyhedron, so among all $\lambda_0 \geq 0$ and $\boldsymbol{\lambda} \in (\mathbb{R}^+)^L$, most of the inequalities in \eqref{eq:sup-A3} are \emph{redundant}. Identifying those which define the \emph{faces} of the superposition coding rate region $\mathcal{R}_{sup}$ appears to be very difficult. Note, however, that for any given $(R_0,R_1,\ldots,R_L) \in (\mathbb{R}^+)^{L+1}$ and $\boldsymbol{\lambda} \in (\mathbb{R}^+)^L$, the left-hand side of \eqref{eq:sup-A3} is a linear, nondecreasing function of $\lambda_0$, and the right-hand side of \eqref{eq:sup-A3} is a piecewise linear, nondecreasing, and concave function of $\lambda_0$. Thus, the left-hand side of \eqref{eq:sup-A3} will dominate the right-hand side for \emph{every} $\lambda_0 \geq 0$ if and only if it dominates the right-hand side at its boundary points $\lambda_0=f_m(\boldsymbol{\lambda})$, $m=1,\ldots,L$, between the adjacent line segments. See Figure~\ref{fig:obs} for an illustration.

\begin{figure}[!t]
\centering
\includegraphics[width=0.75\linewidth,draft=false]{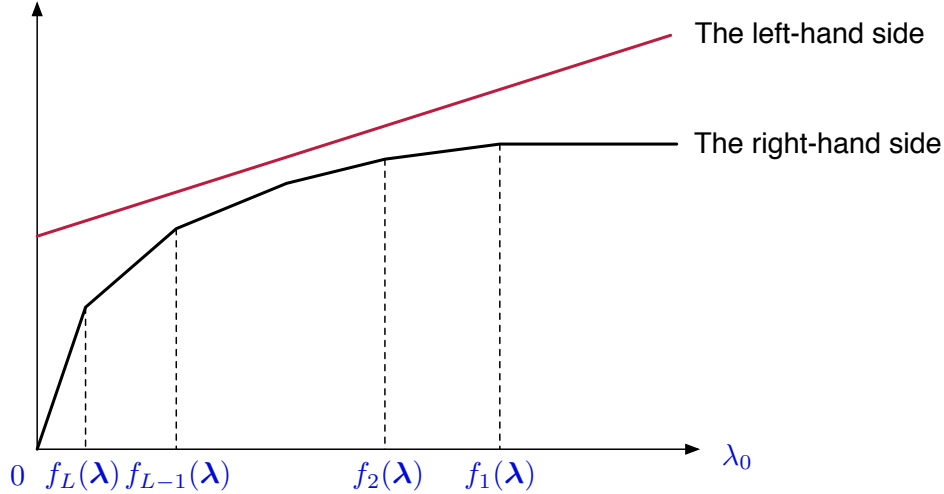}
\caption{The left-hand and right-hand sides of \eqref{eq:sup-A3} as a function of $\lambda_0$ for a fixed $(R_0,R_1,\ldots,R_L) \in (\mathbb{R}^+)^{L+1}$ and $\boldsymbol{\lambda} \in (\mathbb{R}^+)^L$.}
\label{fig:obs}
\end{figure}

Formally, we have the following proposition, which plays a key role next in proving the optimality of superposition coding for achieving the entire admissible rate region of SMDC-A.

\begin{prop}\label{prop:SC}
The superposition coding rate region $\mathcal{R}_{sup}$ for SMDC-A is given by the collection of all nonnegative rate tuples $(R_0,R_1,\ldots,R_L)$ satisfying
\begin{align}
f_m(\boldsymbol{\lambda})R_0+\sum_{l=1}^L\lambda_lR_l & \geq  \sum_{\alpha=1}^L \min(f_\alpha(\boldsymbol{\lambda}),f_m(\boldsymbol{\lambda}))H(S_\alpha) \label{eq:sup-A4}\\
&= f_m(\boldsymbol{\lambda})\sum_{\alpha=1}^{m}H(S_\alpha)+\sum_{\alpha=m+1}^Lf_\alpha(\boldsymbol{\lambda})H(S_\alpha), \quad \forall m=1,\ldots,L\; \mbox{and} \; \boldsymbol{\lambda} \in (\mathbb{R}^+)^L
\label{eq:sup-A5}
\end{align}
where $f_\alpha(\boldsymbol{\lambda})$ is the optimal value of the linear program \eqref{eq:dual2}.
\end{prop}

\begin{proof}
Let us first recall the following results from \cite{YZ-IT99}: for any $\boldsymbol{\lambda} \in(\mathbb{R}^+)^L$ we have
\begin{equation}
f_1(\boldsymbol{\lambda}) \geq f_2(\boldsymbol{\lambda}) \geq \cdots \geq f_L(\boldsymbol{\lambda}) \geq 0.
\label{eq:IC2}
\end{equation}
It follows that
\begin{align}
\sum_{\alpha=1}^L \min(f_\alpha(\boldsymbol{\lambda}),f_m(\boldsymbol{\lambda}))H(S_\alpha) &=
\sum_{\alpha=1}^m f_m(\boldsymbol{\lambda})H(S_\alpha)+\sum_{\alpha=m+1}^L f_\alpha(\boldsymbol{\lambda})H(S_\alpha)\\
&=
f_m(\boldsymbol{\lambda})\sum_{\alpha=1}^m H(S_\alpha)+\sum_{\alpha=m+1}^L f_\alpha(\boldsymbol{\lambda})H(S_\alpha).
\end{align}
It remains to show that for any given $\boldsymbol{\lambda} \in (\mathbb{R}^+)^L$, the set of inequalities \eqref{eq:sup-A3} over all $\lambda_0 \geq 0$ is dominated by that over $\lambda_0=f_m(\boldsymbol{\lambda})$ for $m=1,\ldots,L$. 

Fix $\boldsymbol{\lambda} \in (\mathbb{R}^+)^L$, and consider the following three cases separately.

Case 1: $\lambda_0\geq f_1(\boldsymbol{\lambda})$. By \eqref{eq:IC2}, we have $\lambda_0\geq f_\alpha(\boldsymbol{\lambda})$ and hence $\min(f_\alpha(\boldsymbol{\lambda}),\lambda_0)=f_\alpha(\boldsymbol{\lambda})$ for any $\alpha=1,\ldots, L$. For $m=1$, the inequality \eqref{eq:sup-A5} can be written as
\begin{equation}
f_1(\boldsymbol{\lambda})R_0+\sum_{l=1}^{L}\lambda_lR_l \geq \sum_{\alpha=1}^{L}f_\alpha (\boldsymbol{\lambda})H(S_\alpha)
\label{Pf-SC-5}
\end{equation}
which implies that
\begin{align}
\lambda_0R_0+\sum_{l=1}^{L}\lambda_lR_l & \geq f_1(\boldsymbol{\lambda})R_0+\sum_{l=1}^{L}\lambda_lR_l\\
& \geq \sum_{\alpha=1}^{L}f_\alpha(\boldsymbol{\lambda})H(S_\alpha)\\
&= \sum_{\alpha=1}^L \min(f_\alpha(\boldsymbol{\lambda}),\lambda_0)H(S_\alpha)
\end{align}
for any $\lambda_0\geq f_1(\boldsymbol{\lambda})$.

Case 2: $0 \leq \lambda_0 < f_L(\boldsymbol{\lambda})$. By \eqref{eq:IC2}, we have $\lambda_0 < f_\alpha(\boldsymbol{\lambda})$ and hence $\min(f_\alpha(\boldsymbol{\lambda}),\lambda_0)=\lambda_0$ for any $\alpha=1,\ldots, L$. For $m=L$, the inequality \eqref{eq:sup-A5} can be written as
\begin{equation}
f_L(\boldsymbol{\lambda})R_0+\sum_{l=1}^{L}\lambda_lR_l \geq f_L(\boldsymbol{\lambda})\sum_{\alpha=1}^{L}H(S_\alpha)
\label{Pf-SC-6}
\end{equation}
which implies that
\begin{eqnarray}
\lambda_0R_0+\sum_{l=1}^{L}\lambda_lR_l 
& \geq &  \frac{\lambda_0}{f_L(\boldsymbol{\lambda})}\left(f_L(\boldsymbol{\lambda})R_0+\sum_{l=1}^{L}\lambda_lR_l\right)\\
& \geq & \frac{\lambda_0}{f_L(\boldsymbol{\lambda})}\left(f_L(\boldsymbol{\lambda})\sum_{\alpha=0}^{L}H(S_\alpha)\right)\\
& = & \lambda_0 \sum_{\alpha=0}^{L}H(S_\alpha)\\
&= & \sum_{\alpha=1}^L \min(f_\alpha(\boldsymbol{\lambda}),\lambda_0)H(S_\alpha)
\end{eqnarray}
for any $0 \leq \lambda_0 < f_L(\boldsymbol{\lambda})$.

Case 3: $f_{r+1}(\boldsymbol{\lambda}) \leq \lambda_0 < f_r(\boldsymbol{\lambda})$ for some $r=1,\ldots,L-1$. By \eqref{eq:IC2}, we have $\lambda_0 < f_\alpha(\boldsymbol{\lambda})$ and hence $\min(f_\alpha(\boldsymbol{\lambda}),\lambda_0)=\lambda_0$ for $\alpha=1,\ldots, r$, and $\lambda_0 \geq f_\alpha(\boldsymbol{\lambda})$ and hence $\min(f_\alpha(\boldsymbol{\lambda}),\lambda_0)=f_\alpha(\boldsymbol{\lambda})$ for $\alpha=r+1,\ldots, L$. For $m=r$ and $r+1$, the inequality \eqref{eq:sup-A5} can be written as
\begin{eqnarray}
f_r(\boldsymbol{\lambda})R_0+\sum_{l=1}^{L}\lambda_lR_l & \geq&  f_r(\boldsymbol{\lambda})\sum_{\alpha=1}^{r}H(S_\alpha)+\sum_{\alpha=r+1}^{L}f_\alpha(\boldsymbol{\lambda})H(S_\alpha)\label{Pf-SC-7}\\
\mbox{and} \quad f_{r+1}(\boldsymbol{\lambda})R_0+\sum_{l=1}^{L}\lambda_lR_l & \geq&  f_{r+1}(\boldsymbol{\lambda})\sum_{\alpha=1}^{r+1}H(S_\alpha)+\sum_{\alpha=r+2}^{L}f_\alpha(\boldsymbol{\lambda})H(S_\alpha)\label{Pf-SC-8}
\end{eqnarray}
respectively, which together imply that
\begin{align}
\lambda_0R_0+\sum_{l=1}^{L}\lambda_lR_l & = \frac{\lambda_0-f_{r+1}(\boldsymbol{\lambda})}{f_r(\boldsymbol{\lambda})-f_{r+1}(\boldsymbol{\lambda})}\left(f_r(\boldsymbol{\lambda})R_0+\sum_{l=1}^{L}\lambda_lR_l\right)+\notag\\
& \hspace{16pt} \frac{f_r(\boldsymbol{\lambda})-\lambda_0}{f_r(\boldsymbol{\lambda})-f_{r+1}(\boldsymbol{\lambda})}\left(f_{r+1}(\boldsymbol{\lambda})R_0+\sum_{l=1}^{L}\lambda_lR_l\right)\\
& \geq \frac{\lambda_0-f_{r+1}(\boldsymbol{\lambda})}{f_r(\boldsymbol{\lambda})-f_{r+1}(\boldsymbol{\lambda})}\left(f_r(\boldsymbol{\lambda})\sum_{\alpha=1}^{r}H(S_\alpha)+\sum_{\alpha=r+1}^{L}f_\alpha(\boldsymbol{\lambda})H(S_\alpha)\right)+\notag\\
& \hspace{16pt} \frac{f_r(\boldsymbol{\lambda})-\lambda_0}{f_r(\boldsymbol{\lambda})-f_{r+1}(\boldsymbol{\lambda})}\left(f_{r+1}(\boldsymbol{\lambda})\sum_{\alpha=1}^{r+1}H(S_\alpha)+\sum_{\alpha=r+2}^{L}f_\alpha(\boldsymbol{\lambda})H(S_\alpha)\right)\\
&= \lambda_0\sum_{\alpha=1}^{r}H(S_\alpha)+\sum_{\alpha=r+1}^{L}f_\alpha(\boldsymbol{\lambda})H(S_\alpha)\\
&=  \sum_{\alpha=1}^L \min(f_\alpha(\boldsymbol{\lambda}),\lambda_0)H(S_\alpha)
\end{align}
for any $f_{r+1}(\boldsymbol{\lambda}) \leq \lambda_0 < f_r(\boldsymbol{\lambda})$.

Combining these three cases completes the proof of the proposition.
\end{proof}

\subsubsection{Optimality of Superposition Coding}\label{sec:Ext-SMDCA-Main}
The main result of this section is that superposition coding remains optimal in terms of achieving the entire admissible rate region for SMDC-A, as summarized in the following theorem.

\begin{theorem}\label{thm:SMDC-A}
For the general SMDC-A problem, the admissible rate region
\begin{equation}
\mathcal{R}=\mathcal{R}_{sup}.
\end{equation}
\end{theorem}

\begin{proof}
Based on the discussions from Section~\ref{sec:Ext-SMDCA-SC}, we naturally have $\mathcal{R}_{sup} \subseteq \mathcal{R}$. Thus, to show $\mathcal{R}_{sup} = \mathcal{R}$ we only need to show that $\mathcal{R} \subseteq \mathcal{R}_{sup}$. In light of Proposition~\ref{prop:SC}, it is sufficient to show that \emph{any} admissible rate tuple $(R_0,R_1,\ldots,R_L)$ must satisfy
\begin{align}
f_m(\boldsymbol{\lambda})R_0+\sum_{l=1}^L\lambda_lR_l & \geq  f_m(\boldsymbol{\lambda})\sum_{\alpha=1}^{m}H(S_\alpha)+\sum_{\alpha=m+1}^Lf_\alpha(\boldsymbol{\lambda})H(S_\alpha)
\end{align}
for \emph{all} $m=1,\ldots,L$ and $\boldsymbol{\lambda} \in (\mathbb{R}^+)^L$.

Let $(R_0,R_1,\ldots,R_L)$ be an admissible rate tuple. By definition, for any sufficiently large block-length $n$ there exists an 
$(n,(M_0,M_1,\ldots,M_L))$ code satisfying the rate constraints \eqref{eq:Rate-A} for the admissible rate tuple $(R_0,R_1,\ldots,R_L)$ and the asymptotically perfect reconstruction requirement \eqref{eq:PR-A}. Fix $\boldsymbol{\lambda} \in (\mathbb{R}^+)^L$, and let $\{c_{\boldsymbol{\lambda}}^{(\alpha)}:\alpha=1,\ldots,L\}$ be a set of \emph{optimal} solutions that satisfies the subset entropy inequality of Yeung and Zhang \eqref{eq:YZ}. 

Note that for $\alpha=1$, the optimal solution for the linear program \eqref{eq:dual2} is \emph{unique} and is given by
\begin{equation}
c_{\boldsymbol{\lambda}}(\{l\})=\lambda_l, \quad \forall l=1,\ldots,L.
\end{equation}
We thus have for any $m=1,\ldots,L$
\begin{align}
n\left(f_m(\boldsymbol{\lambda})R_0+\sum_{l=1}^L\lambda_lR_l\right)
&= f_m(\boldsymbol{\lambda})nR_0+\sum_{l=1}^Lc_{\boldsymbol{\lambda}}(\{l\})nR_l\\
& \geq  f_m(\boldsymbol{\lambda})(H(X_0)-n\epsilon)+\sum_{l=1}^Lc_{\boldsymbol{\lambda}}(\{l\})(H(X_l)-n\epsilon) \label{Pf-Main-1}\\
& =  f_m(\boldsymbol{\lambda})H(X_0)+\sum_{l=1}^Lc_{\boldsymbol{\lambda}}(\{l\})H(X_l)-n(f_1(\boldsymbol{\lambda})+f_m(\boldsymbol{\lambda}))\epsilon \label{Pf-Main-2}\\
& \geq  f_m(\boldsymbol{\lambda})H(X_0)+\sum_{U \in \Omega_L^{(m)}}c_{\boldsymbol{\lambda}}(U)H(X_U)-n(f_1(\boldsymbol{\lambda})+f_m(\boldsymbol{\lambda}))\epsilon \label{Pf-Main-3}\\
& =  \sum_{U \in \Omega_L^{(m)}}c_{\boldsymbol{\lambda}}(U)(H(X_0)+H(X_U))-n(f_1(\boldsymbol{\lambda})+f_m(\boldsymbol{\lambda}))\epsilon \label{Pf-Main-4}\\
& \geq  \sum_{U \in \Omega_L^{(m)}}c_{\boldsymbol{\lambda}}(U)H(X_0,X_U)-n(f_1(\boldsymbol{\lambda})+f_m(\boldsymbol{\lambda}))\epsilon \label{Pf-Main-5}
\end{align}
where \eqref{Pf-Main-1} follows from the rate constraint \eqref{eq:Rate-A}, \eqref{Pf-Main-2} and \eqref{Pf-Main-4} are due to the fact that $c_{\boldsymbol{\lambda}}^{(1)}$ and $c_{\boldsymbol{\lambda}}^{(m)}$ are optimal so we have 
\begin{align}
\sum_{l=1}^Lc_{\boldsymbol{\lambda}}(\{l\}) =f_1({\boldsymbol{\lambda}}) \quad
\mbox{and} \quad \sum_{U \in \Omega_L^{(m)}}c_{\boldsymbol{\lambda}}(U) = f_m({\boldsymbol{\lambda}})
\end{align}
\eqref{Pf-Main-3} follows from the subset entropy inequality of Yeurng and Zhang \eqref{eq:YZ} so we have
\begin{equation}
\sum_{l=1}^Lc_{\boldsymbol{\lambda}}(\{l\})H(X_l) \geq \sum_{U \in \Omega_L^{(m)}}c_{\boldsymbol{\lambda}}(U)H(X_U)
\end{equation}
and \eqref{Pf-Main-5} follows from the independence bound on entropy. 

For any $U \in \Omega_L^{(m)}$ and $m=1,\ldots,L$, by the asymptotically perfect reconstruction requirement \eqref{eq:PR-A} and the well-known Fano's inequality we have
\begin{equation}
H(S_1^n,\ldots,S_m^n|X_0,X_U) \leq n\delta_m^{(n)}
\label{Pf-Main-6}
\end{equation}
where $\delta_m^{(n)} \rightarrow 0$ in the limit as $n \rightarrow \infty$ and $\epsilon \rightarrow 0$. By the chain rule for entropy, 
\begin{eqnarray}
H(X_0,X_U) &= & H(X_0,X_U,S_1^n,\ldots,S_m^n)-H(S_1^n,\ldots,S_m^n|X_0,X_U)\label{Pf-Main-8}\\
&= & H(S_1^n,\ldots,S_m^n)+H(X_0,X_U|S_1^n,\ldots,S_m^n)-H(S_1^n,\ldots,S_m^n|X_0,X_U)\label{Pf-Main-9}\\
&=& n\sum_{\alpha=1}^mH(S_\alpha)+H(X_0,X_U|S_1^n,\ldots,S_m^n)-H(S_1^n,\ldots,S_m^n|X_0,X_U)\label{Pf-Main-10}\\
&\geq & n\sum_{\alpha=1}^mH(S_\alpha)+H(X_0,X_U|S_1^n,\ldots,S_m^n)-n\delta_m^{(n)}\label{Pf-Main-11}
\end{eqnarray}
where \eqref{Pf-Main-10} is due to the fact that $S_1,\ldots,S_L$ are independent memoryless sources, and \eqref{Pf-Main-11} follows from \eqref{Pf-Main-6}. Substituting \eqref{Pf-Main-11} into \eqref{Pf-Main-5}, we have
\begin{align}
&n\left(f_m(\boldsymbol{\lambda})R_0+\sum_{l=1}^L\lambda_lR_l\right)\notag\\
& \geq \sum_{U \in \Omega_L^{(m)}}c_{\boldsymbol{\lambda}}(U)\left(
n\sum_{\alpha=1}^mH(S_\alpha)+H(X_0,X_U|S_1^n,\ldots,S_m^n)-n\delta_m^{(n)}\right)-n(f_1(\boldsymbol{\lambda})+f_m(\boldsymbol{\lambda}))\epsilon\\
&= nf_m(\boldsymbol{\lambda})\sum_{\alpha=1}^mH(S_\alpha)+\sum_{U \in \Omega_L^{(m)}}c_{\boldsymbol{\lambda}}(U)H(X_0,X_U|S_1^n,\ldots,S_m^n)-\notag\\
&\hspace{15pt} n(f_m(\boldsymbol{\lambda})\delta_m^{(n)}+(f_1(\boldsymbol{\lambda})+f_m(\boldsymbol{\lambda}))\epsilon).
\label{Pf-Main-11.5}
\end{align}

Next, we show, via an induction on $m$, that for any $m=1,\ldots,L$ we have
\begin{equation}
\sum_{U \in \Omega_L^{(m)}}c_{\boldsymbol{\lambda}}(U)H(X_0,X_U|S_1^n,\ldots,S_m^n) \geq n
\left(\sum_{\alpha=m+1}^{L}f_\alpha(\boldsymbol{\lambda}) H(S_\alpha)-\sum_{\alpha=m+1}^{L}f_\alpha(\boldsymbol{\lambda})\delta_\alpha^{(n)}\right).
\label{Pf-Main-12}
\end{equation}
First consider the base case with $m=L$. In this case, the inequality \eqref{Pf-Main-12} is trivial as the right-hand side of the inequality is zero. Next, assume that the inequality \eqref{Pf-Main-12} holds for $m=l$ for some $l=2,\ldots,L$, i.e,
\begin{equation}
\sum_{U \in \Omega_L^{(l)}}c_{\boldsymbol{\lambda}}(U)H(X_0,X_U|S_1^n,\ldots,S_l^n) \geq n
\left(\sum_{\alpha=l+1}^{L}f_\alpha(\boldsymbol{\lambda}) H(S_\alpha)-\sum_{\alpha=l+1}^{L}f_\alpha(\boldsymbol{\lambda})\delta_\alpha^{(n)}\right).
\label{Pf-Main-13}
\end{equation}
For any $U \in \Omega_L^{(l)}$, we have 
\begin{align}
H&(X_0,X_U|S_1^n,\ldots,S_l^n)\notag\\
&=  H(X_0,X_U|S_1^n,\ldots,S_{l-1}^n)-I(S_l^n;X_0,X_U|S_1^n,\ldots,S_{l-1}^n)\\
&= H(X_0,X_U|S_1^n,\ldots,S_{l-1}^n)-H(S_l^n|S_1^n,\ldots,S_{l-1}^n)+H(S_l^n|X_0,X_U,S_1^n,\ldots,S_{l-1}^n)\\
& \leq H(X_0,X_U|S_1^n,\ldots,S_{l-1}^n)-H(S_l^n|S_1^n,\ldots,S_{l-1}^n)+H(S_l^n|X_0,X_U)\label{Pf-Main-14}\\
& \leq H(X_0,X_U|S_1^n,\ldots,S_{l-1}^n)-H(S_l^n|S_1^n,\ldots,S_{l-1}^n)+\delta_l^{(n)}\label{Pf-Main-15}\\
& =  H(X_0,X_U|S_1^n,\ldots,S_{l-1}^n)-nH(S_l)+\delta_l^{(n)}\label{Pf-Main-16}
\end{align}
where \eqref{Pf-Main-14} follows from the fact that conditioning reduces entropy, \eqref{Pf-Main-15} follows the fact that
\begin{equation}
H(S_l^n|X_0,X_U) \leq H(S_1^n,\ldots,S_l^n|X_0,X_U) \leq n\delta_l^{(n)}
\end{equation}
and \eqref{Pf-Main-16} follows from the fact that $S_1,\ldots,S_L$ are independent memoryless sources. Multiplying both sides of the inequality \eqref{Pf-Main-16} by $c_{\boldsymbol{\lambda}}(U)$ and summing over all $U \in \Omega_L^{(l)}$, we have
\begin{align}
\sum_{U \in \Omega_L^{(l)}}&c_{\boldsymbol{\lambda}}(U)H(X_0,X_U|S_1^n,\ldots,S_{l-1}^n)\notag\\
& \geq \sum_{U \in \Omega_L^{(l)}}c_{\boldsymbol{\lambda}}(U)\left(H(X_0,X_U|S_1^n,\ldots,S_l^n)+nH(S_l)-n\delta_l^{(n)}\right)\\
& = \sum_{U \in \Omega_L^{(l)}}c_{\boldsymbol{\lambda}}(U)H(X_0,X_U|S_1^n,\ldots,S_l^n)+n\left(f_l(\boldsymbol{\lambda})H(S_l)-f_l(\boldsymbol{\lambda})n\delta_l^{(n)}\right)\\
& \geq n\left(\sum_{\alpha=l+1}^{L}f_\alpha(\boldsymbol{\lambda})H(S_\alpha)-\sum_{\alpha=l+1}^{L}f_\alpha(\boldsymbol{\lambda})\delta_\alpha^{(n)}\right)+n\left(f_l(\boldsymbol{\lambda})H(S_l)-f_l(\boldsymbol{\lambda})\delta_l^{(n)}\right)\label{Pf-Main-17}\\
& = n\left(\sum_{\alpha=l}^{L}f_\alpha(\boldsymbol{\lambda})H(S_\alpha)-\sum_{\alpha=l}^{L}f_\alpha(\boldsymbol{\lambda})\delta_\alpha^{(n)}\right)
\end{align}
where \eqref{Pf-Main-17} follows from the induction assumption \eqref{Pf-Main-13}. Finally, by the subset entropy inequality of Yeung and Zhang \eqref{eq:YZ} we have
\begin{align}
\sum_{U \in \Omega_L^{(l-1)}}c_{\boldsymbol{\lambda}}(U)H(X_0,X_U|S_1^n,\ldots,S_{l-1}^n) & \geq \sum_{U \in \Omega_L^{(l)}}c_{\boldsymbol{\lambda}}(U)H(X_0,X_U|S_1^n,\ldots,S_{l-1}^n)\\
& \geq n\left(\sum_{\alpha=l}^{L}f_\alpha(\boldsymbol{\lambda})H(S_\alpha)-\sum_{\alpha=l}^{L}f_\alpha(\boldsymbol{\lambda})\delta_\alpha^{(n)}\right).
\end{align}
This proves that the inequality \eqref{Pf-Main-12} also holds for $m=l-1$ and hence completes the proof of \eqref{Pf-Main-12}.

Substituting \eqref{Pf-Main-12} into \eqref{Pf-Main-11.5} and dividing both sides of the inequality by $n$, we have
\begin{align}
f_m(\boldsymbol{\lambda})R_0+\sum_{l=1}^L\lambda_lR_l &\geq f_m(\boldsymbol{\lambda})\sum_{\alpha=1}^mH(S_\alpha)+\sum_{\alpha=m+1}^{L}f_\alpha(\boldsymbol{\lambda})H(S_\alpha)-\notag\\
&\hspace{15pt} \left(\sum_{\alpha=m}^{L}f_\alpha(\boldsymbol{\lambda})\delta_\alpha^{(n)}+(f_1(\boldsymbol{\lambda})+f_m(\boldsymbol{\lambda}))\epsilon\right). \label{Pf-Main-18}
\end{align}
Letting $n \rightarrow \infty$ and $\epsilon \rightarrow 0$, we have from \eqref{Pf-Main-18} that 
\begin{equation}
f_m(\boldsymbol{\lambda})R_0+\sum_{l=1}^L\lambda_lR_l \geq f_m(\boldsymbol{\lambda})\sum_{\alpha=1}^mH(S_\alpha)+\sum_{\alpha=m+1}^{L}f_\alpha(\boldsymbol{\lambda})H(S_\alpha)
\end{equation}
for any admissible rate tuple $(R_0,R_1,\ldots,R_L)$. This proves that $\mathcal{R} \subseteq \mathcal{R}_{sup}$ and hence completes the proof of the theorem.
\end{proof}

\subsubsection{Rate Allocation at the All-Access Encoder}\label{sec:Ext-SMDCA-RA}
In this section, we conclude our discussion on SMDC-A by focusing on a \emph{greedy} rate allocation policy at the all-access encoder. Based on our previous discussion in Section~\ref{sec:Ext-SMDCA-SC}, the output of the all-access encoder 0 consists of only \emph{uncoded} information bits for the source messages $W_1,\ldots,W_L$. Hence, its storage efficiency is the same for each of the information sources $S_1,\ldots,S_L$. On the other hand, for the randomly accessible encoders 1 to $L$, $S_1$ has the highest reconstruction requirement and hence is the least efficient source to encode, and $S_L$ has the lowest reconstruction requirement and hence is the most efficient source to encode. Therefore, intuitively, the greedy policy that assigns the remaining rate budget of the all-access encoder 0 to the least efficient source should be optimal. 

More specifically, suppose that the rate budget $R_0$ of the all-access encoder 0 satisfies
\begin{align}
\sum_{\alpha=1}^{q-1}H(S_\alpha) \leq R_0 < \sum_{\alpha=1}^{q}H(S_\alpha)
\end{align}
for some $q=1,\ldots,L$. The greedy policy stores the source messages $W_1,\ldots,W_{q-1}$ in their entireties (without any coding) at the all-access encoder 0, and the residual rate budget $R_0-\sum_{\alpha=1}^{q-1}H(S_\alpha)$ is then committed in full to the source message $W_q$. The residual source messages are $W_q$, with a residual rate
\begin{align}
H(S_q)-\left(R_0-\sum_{\alpha=1}^{q-1}H(S_\alpha)\right)=\sum_{\alpha=1}^{q}H(S_\alpha)-R_0
\end{align}
and $W_{q+1},\ldots,W_L$ with respective rates $H(S_{q+1}),\ldots,H(S_L)$. The residual source messages are encoded at the randomly accessible encoders using superposition coding, and the corresponding rate region $\mathcal{R}_{sup}'(R_0)$ is given by
\begin{align}
\mathcal{R}_{sup}'(R_0) &= \left\{(R_1,\ldots,R_L)\in (\mathbb{R}^+)^L: \sum_{l=1}^L\lambda_lR_l \geq  f_q(\boldsymbol{\lambda})\left(\sum_{\alpha=1}^{q}H(S_\alpha)-R_0\right)+\right.\notag\\
&\hspace{200pt} \left.\sum_{\alpha=q+1}^Lf_\alpha(\boldsymbol{\lambda})H(S_\alpha), \quad \forall \boldsymbol{\lambda} \in (\mathbb{R}^+)^L\right\}.
\end{align}
Of course, when
\begin{align}
R_0 \geq \sum_{\alpha=1}^{L}H(S_\alpha)
\end{align}
all source messages $W_1,\ldots,W_L$ can be stored at the all-access encoder 0 (without any coding), and there is no need to use the randomly access encoders 1 to $L$. In this case, we have $\mathcal{R}_{sup}'(R_0)=(\mathbb{R}^+)^L$.

To show that the aforementioned greedy rate allocation policy at the all-access encoder 0 is optimal, we need to show that $\mathcal{R}_{sup}'(R_0)$ \emph{matches} the $R_0$-slice of the superposition coding rate region 
\begin{align}
\mathcal{R}_{sup}(R_0) := \left\{(R_1,\ldots,R_L)\in (\mathbb{R}^+)^L:(R_0,R_1,\ldots,R_L) \in \mathcal{R}_{sup}\right\}
\end{align}
for \emph{all} $R_0 \geq 0$. By Proposition~\ref{prop:SC}, for any $R_0 \geq 0$ the $R_0$-slice of the superposition coding rate region can be written as
\begin{align}
\mathcal{R}_{sup}(R_0) &= \left\{(R_1,\ldots,R_L)\in (\mathbb{R}^+)^L: \sum_{l=1}^L\lambda_lR_l \geq  \max_{m=1,\ldots,L}\left\{
g_m(\boldsymbol{\lambda})\right\}, \quad \forall \boldsymbol{\lambda} \in (\mathbb{R}^+)^L\right\}
\end{align}
where
\begin{align}
g_m(\boldsymbol{\lambda}) := f_m(\boldsymbol{\lambda})\left(\sum_{\alpha=1}^{m}H(S_\alpha)-R_0\right)+
\sum_{\alpha=m+1}^Lf_\alpha(\boldsymbol{\lambda})H(S_\alpha)
\end{align}
For any $m=1,\ldots,L-1$, it is straightforward to calculate that
\begin{align}
g_{m+1}(\boldsymbol{\lambda})-g_m(\boldsymbol{\lambda}) = \left(f_{m+1}(\boldsymbol{\lambda})-f_{m}(\boldsymbol{\lambda})\right)\left(\sum_{\alpha=1}^{m}H(S_\alpha)-R_0\right).
\end{align}
By \eqref{eq:IC2}, $f_{m+1}(\boldsymbol{\lambda})-f_{m}(\boldsymbol{\lambda}) \leq 0$ for any $m=1,\ldots,L-1$ and $\boldsymbol{\lambda} \in (\mathbb{R}^+)^L$. Thus, when $\sum_{\alpha=1}^{q-1}H(S_\alpha) \leq R_0 < \sum_{\alpha=1}^{q}H(S_\alpha)$ for some $q=1,\ldots,L$, we have $\sum_{\alpha=1}^{m}H(S_\alpha)-R_0 \geq 0$ and hence $g_{m+1}(\boldsymbol{\lambda})-g_m(\boldsymbol{\lambda}) \leq 0$ for all $m=q,\ldots,L-1$, and $\sum_{\alpha=1}^{m}H(S_\alpha)-R_0 \leq 0$ and hence $g_{m+1}(\boldsymbol{\lambda})-g_m(\boldsymbol{\lambda}) \geq 0$ for all $m=1,\ldots,q-1$. We conclude that in this case,
\begin{align}
\max_{m=1,\ldots,L}\left\{g_m(\boldsymbol{\lambda})\right\}=g_q(\boldsymbol{\lambda})=f_q(\boldsymbol{\lambda})\left(\sum_{\alpha=1}^{q}H(S_\alpha)-R_0\right)+\sum_{\alpha=q+1}^Lf_\alpha(\boldsymbol{\lambda})H(S_\alpha)
\end{align}
for any $\boldsymbol{\lambda} \in (\mathbb{R}^+)^L$ and hence $\mathcal{R}_{sup}(R_0)=\mathcal{R}_{sup}'(R_0)$. When$R_0 \geq \sum_{\alpha=1}^{L}H(S_\alpha)$, we have $\sum_{\alpha=1}^{m}H(S_\alpha)-R_0 \leq 0$ and hence $g_{m+1}(\boldsymbol{\lambda})-g_m(\boldsymbol{\lambda}) \geq 0$ for all $m=1,\ldots,L-1$. In this case,
\begin{align}
\max_{m=1,\ldots,L}\left\{g_m(\boldsymbol{\lambda})\right\}=g_L(\boldsymbol{\lambda})=f_L(\boldsymbol{\lambda})\left(\sum_{\alpha=1}^{L}H(S_\alpha)-R_0\right)\leq 0
\end{align}
for any $\boldsymbol{\lambda} \in (\mathbb{R}^+)^L$, and we once again have $\mathcal{R}_{sup}(R_0)=\mathcal{R}_{sup}'(R_0)$. We summarize the above results in the following theorem.

\begin{theorem}
Greedy rate allocation at the all-access encoder combined with superposition coding at the randomly accessible encoders can achieve the entire admissible rate region for the general SMDC-A problem.
\end{theorem}

\subsection{Extension 2: S-SMDC}\label{sec:Ext-SSMDC}

\subsubsection{Problem Statement}\label{sec:Ext-SSMDC-PS}
\begin{figure}[!t]
\centering
\includegraphics[width=0.95\linewidth,draft=false]{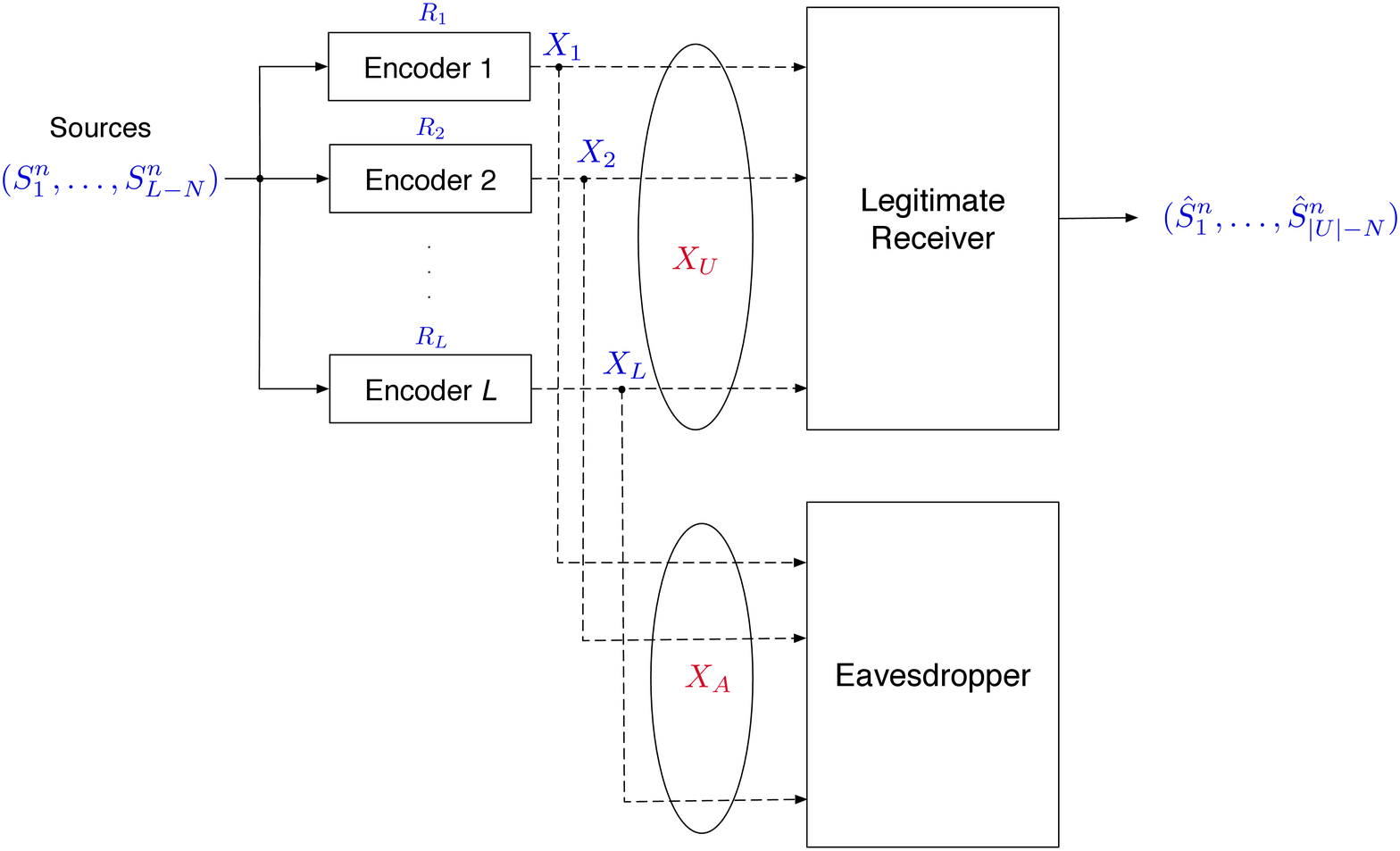}
\caption{S-SMDC with $L$ randomly accessible encoders $1$ to $L$. A total of $L-N$ independent discrete memoryless sources $(S_1,\ldots,S_{L-N})$ are to be encoded at the encoders. The legitimate receiver, which has access to a subset $U$ of the encoder outputs, needs to nearly perfectly reconstruct the sources $(S_1,\ldots,S_{|U|-N})$ whenever $|U| \geq N+1$. The eavesdropper has access to a subset $A$ of the encoder ouputs. All sources $(S_1,\ldots,S_{L-N})$ need to be kept perfectly secret from the eavesdropper whenever $|A| \leq N$.}
\label{fig:SSMDC}
\end{figure}

Let $L$ be a positive integer, and let $N \in \{0,\ldots,L-1\}$. Let $\{S_1[t],\ldots,S_{L-N}[t]\}_{t=1}^{\infty}$ be a collection of $L-N$ \emph{independent} discrete memoryless sources with time index $t$, and let $S_\alpha^n:=(S_\alpha[1],\ldots,S_\alpha[n])$ for $\alpha=1,\ldots,L-N$. As illustrated in Figure~\ref{fig:SSMDC}, an $(L,N)$ S-SMDC problem consists of a set of $L$ encoders, a legitimate receiver who has access to a subset $U$ of the encoder outputs, and an eavesdropper who has access to a subset $A$ of the encoder outputs. Which subsets of the encoder outputs are available at the legitimate receiver and the eavesdropper are \emph{unknown} a priori at the encoders. However, no matter which subsets $U$ and $A$ actually occur, the legitimate receiver must be able to asymptotically perfectly reconstruct the sources $(S_1,\ldots,S_\alpha)$ whenever $|U| \geq N+\alpha$, and all sources $(S_1,\ldots,S_{L-N})$ must be kept perfectly secure from the eavesdropper as long as $|A|\leq N$. 

Formally, an $(n,(M_1,\ldots,M_L))$ code is defined by a collection of $L$ encoding functions
\begin{equation}
e_l: \prod_{\alpha=1}^{L-N}\mathcal{S}_\alpha^n \times \mathcal{K} \rightarrow \{1,\ldots,M_l\}, \quad \forall l=1,\ldots,L
\end{equation}
and $\sum_{\alpha=N+1}^L
\left(
\begin{array}{c}
L   \\
\alpha
\end{array}
\right)$ decoding functions
\begin{equation}
d_U: \prod_{l \in U}\{1,\ldots,M_l\} \rightarrow \prod_{\alpha=1}^{|U|-N}\mathcal{S}_\alpha^n, \quad \forall U \subseteq \Omega_L \; \mbox{s.t.} \; |U| \geq N+1
\end{equation}
where $\mathcal{K}$ is the key space accessible to all $L$ encoders. A nonnegative rate tuple $(R_1,\ldots,R_L)$ is said to be \emph{admissible} if for every $\epsilon>0$, there exits, for sufficiently large block length $n$, an $(n,(M_1,\ldots,M_L))$ code such that:
\begin{itemize}
\item (Rate constraints)
\begin{equation}
\frac{1}{n}\log M_l \leq R_l +\epsilon, \quad \forall l =1,\ldots,L;
\label{eq:M-Rate}
\end{equation}
\item (Asymptotically perfect reconstruction at the legitimate receiver)
\begin{equation}
\mathrm{Pr}\{d_U(X_U) \neq (S_1^n,\ldots,S_{|U|-N}^n)\} \leq \epsilon, \quad \forall U \subseteq \Omega_L \; \mbox{s.t.} \; |U| \geq N+1
\label{eq:M-Rec}
\end{equation}
where $X_l:=e_l((S_1^n,\ldots,S_{L-N}^n),K)$ is the output of the $l$th encoder, and $K$ is the secret key shared by all $L$ encoders; and
\item (Perfect secrecy at the eavesdropper)
\begin{equation}
H(S_1^n,\ldots,S_{L-N}^n|X_A) = H(S_1^n,\ldots,S_{L-N}^n), \quad \forall A \subseteq \Omega_L \; \mbox{s.t.} \; |A|\leq N
\label{eq:M-Per}
\end{equation}
i.e., observing the encoder outputs $X_A$ does not provide \emph{any} information regarding to the sources $(S_1^n,\ldots,S_{L-N}^n)$. 
\end{itemize}
The \emph{admissible rate region} $\mathcal R$ is the collection of \emph{all} admissible rate tuples $(R_1,\dots,R_L)$.

\subsubsection{Superposition Coding Rate Region}\label{sec:Ext-SSMDC-SC}
A simple strategy for S-SMDC is to encode each of the $L-N$ sources separately without coding across different sources. Formally, the problem of encoding a single source $S_\alpha$ can be viewed as a special case of the general  S-SMDC problem with $H(S_m)=0$ for all $m \neq \alpha$. When $\alpha=1$, the problem of encoding the single source $S_1$ is the well-known $(L,N+1)$ \emph{threshold secret sharing} problem, for which the admissible rate region was characterized in the classical works \cite{Sha-CACM79,Bla-NCC79}. For the general case with $\alpha \geq 1$, the admissible rate region for encoding the single source $S_\alpha$ was characterized in \cite{BLLLM-ITR12} via a connection to the problem of threshold \emph{ramp-type} secret sharing \cite{Yam_ITF85,BM-Crypto85} and utilizing some basic polyhedral structure of the admissible rate region. The result is summarized in the following proposition. 

\begin{prop} \label{prop:S-SSDC}
Let $\mathcal{R}^{(\alpha)}$ be the collection of all admissible rate tuples for encoding the single source $S_\alpha$. Then, $\mathcal{R}^{(\alpha)}$ is given by the collection of all nonnegative tuples $(r_1^{(\alpha)},\ldots,r_L^{(\alpha)})$ such that
\begin{equation}
\sum_{l\in U}r_l^{(\alpha)}\ge H(S_\alpha), \quad \forall U\in\Omega_L^{(\alpha)}.
\label{eq:S-SSDC}
\end{equation}
\end{prop}

By definition, the superposition coding rate region $\mathcal{R}_{sup}$ for encoding the sources $S_1,\ldots,S_{L-N}$ is given by the collection of nonnegative rate tuples $(R_1,\ldots,R_L)$ such that
\begin{equation}
R_l := \sum_{\alpha=1}^{L-N}r_l^{(\alpha)}, \quad \forall (r_l^{(1)},\ldots,r_l^{(L-N)}) \in \prod_{\alpha=1}^{L-N}\mathcal{R}^{(\alpha)}.
\end{equation}
Note that $\mathcal{R}^{(\alpha)}$ is \emph{identical} to the admissible rate region for encoding the single source $S_\alpha$ in classical SMDC (even though the reconstruction and secrecy requirements are different between these two settings). We thus conclude that the superposition coding rate region $\mathcal{R}_{sup}$ for S-SMDC is given by the collection of nonnegative rate tuples $(R_1,\ldots,R_L)$ satisfying
\begin{equation}
\sum_{l=1}^L\lambda_lR_l \geq \sum_{\alpha=1}^{L-N}f_\alpha(\boldsymbol{\lambda})H(S_\alpha), \quad \forall \boldsymbol{\lambda} \in (\mathbb{R}^+)^L
\label{eq:CJH}
\end{equation}
where $f_\alpha(\boldsymbol{\lambda})$ is the optimal value of the linear program \eqref{eq:dual2}.

\subsubsection{Optimality of Superposition Coding}
In \cite{BLLLM-ITR12}, it was shown that superposition coding can achieve the minimum sum rate for the general S-SMDC problem. The proof was based on the trivial conditional version of the subset entropy inequality of Han. The main result of this section is to show that superposition coding can, in fact, achieve the entire admissible region for the general S-SMDC problem. Our main technical tool is the conditional extension of the subset entropy inequality of Yeung and Zhang proved in Theorem~\ref{thm:CYZ}.

\begin{theorem}\label{thm:S-SMDC}
For the general S-SMDC problem, the admissible rate region
\begin{equation}
\mathcal{R}=\mathcal{R}_{sup}.
\end{equation}
\end{theorem}

\begin{proof}
Based on the discussions from Section~\ref{sec:Ext-SSMDC-SC}, we naturally have $\mathcal{R}_{sup} \subseteq \mathcal{R}$. Thus, to show $\mathcal{R}_{sup} = \mathcal{R}$ we only need to show that $\mathcal{R} \subseteq \mathcal{R}_{sup}$, i.e., \emph{any} admissible rate tuple $(R_1,\ldots,R_L)$ must satisfy \eqref{eq:CJH}.

Let $(R_1,\ldots,R_L)$ be an admissible rate tuple. By definition, for any sufficiently large block-length $n$ there exists an $(n,(M_1,\ldots,M_L))$ code satisfying the rate constraints \eqref{eq:M-Rate} for the admissible rate tuple $(R_1,\ldots,R_L)$, the asymptotically perfect reconstruction requirement \eqref{eq:M-Rec}, and the perfect secrecy requirement \eqref{eq:M-Per}. Fix $\boldsymbol{\lambda} \in (\mathbb{R}^+)^L$, and choose $\mathcal{A}^{(\alpha)}$ and $s_{\boldsymbol{\lambda}}^{(\alpha)}$, $\alpha=1,\ldots,L-N$, to satisfy all the requirement of Theorem~\ref{thm:CYZ}. 

First, let us show that
\begin{equation}
H(X_U|S_1^n,\ldots,S_{\alpha-1}^n,X_A) \geq nH(S_\alpha)-n\delta_n^{(\alpha)}+H(X_U|S_1^n,\ldots,S_\alpha^n,X_A)
\label{eq:ML}
\end{equation}
for any $U \in \Omega_L^{(\alpha)}$, $A \in \mathcal{A}_U$, and $\alpha=1,\ldots,L-N$, where $\delta_n^{(\alpha)} \rightarrow 0$ in the limit as $n \rightarrow \infty$ and $\epsilon \rightarrow 0$.

Fix $U \in \Omega_L^{(\alpha)}$, $A \in \mathcal{A}_U$, and $\alpha=1,\ldots,L-N$. By construction $|U|=\alpha$, $|A|=N$, and $A\cap U =\emptyset$, so we have $|U \cup A|=|U|+|A|=N+\alpha$. By the asymptotically perfect reconstruction requirement \eqref{eq:M-Rec} and the well-known Fano's inequality, we have
\begin{equation}
H(S_1^n,\ldots,S_\alpha^n|X_U,X_A)\leq n\delta_n^{(\alpha)} \label{T660}
\end{equation}
where $\delta_n^{(\alpha)} \rightarrow 0$ in the limit as $n \rightarrow \infty$ and $\epsilon \rightarrow 0$. Furthermore, by the perfect secrecy requirement \eqref{eq:M-Per} we have
\begin{equation}
H(S_1^n,\ldots,S_\alpha^n|X_A)=H(S_1^n,\ldots,S_\alpha^n).\label{T661}
\end{equation}
We thus have
\begin{eqnarray}
&& H(X_U|S_1^n,\ldots,S_{\alpha-1}^n,X_A)+n\delta_n^{(\alpha)}\nonumber\\
&& \hspace{20pt} \geq \; H(X_U|S_1^n,\ldots,S_{\alpha-1}^n,X_A)+H(S_1^n,\ldots,S_\alpha^n|X_U,X_A)
\label{T900}\\
&& \hspace{20pt} \geq \; H(X_U|S_1^n,\ldots,S_{\alpha-1}^n,X_A)+H(S_\alpha^n|S_1^n,\ldots,S_{\alpha-1}^n,X_U,X_A)
\\
&& \hspace{20pt} = \; H(X_U,S_{k}^n|S_1^n,\ldots,S_{\alpha-1}^n,X_A)\\
&& \hspace{20pt} = \; H(S_\alpha^n|S_1^n,\ldots,S_{k-1}^n,X_A)+H(X_V|S_1^n,\ldots,S_\alpha^n,X_A)\\
&& \hspace{20pt} = \; H(S_1^n,\ldots,S_\alpha^n|X_A)-H(S_1^n,\ldots,S_{\alpha-1}^n|X_A)+H(X_U|S_1^n,\ldots,S_\alpha^n,X_A)\\
&& \hspace{20pt} = \; H(S_1^n,\ldots,S_\alpha^n)-H(S_1^n,\ldots,S_{\alpha-1}^n|X_A)+H(X_U|S_1^n,\ldots,S_\alpha^n,X_A) \label{T901}\\
&& \hspace{20pt} \geq \; H(S_1^n,\ldots,S_\alpha^n)-H(S_1^n,\ldots,S_{\alpha-1}^n)+H(X_U|S_1^n,\ldots,S_\alpha^n,X_A)\label{T902}\\
&& \hspace{20pt} = \; H(S_\alpha^n|S_1^n,\ldots,S_{\alpha-1}^n)+H(X_U|S_1^n,\ldots,S_\alpha^n,X_A)\\
&& \hspace{20pt} = \; H(S_\alpha^n)+H(X_U|S_1^n,\ldots,S_\alpha^n,X_A) \label{T903}\\
&& \hspace{20pt} = \; nH(S_\alpha)+H(X_U|S_1^n,\ldots,S_\alpha^n,X_A) \label{T904}
\end{eqnarray}
where \eqref{T900} follows from \eqref{T660}, \eqref{T901} follows from \eqref{T661}, \eqref{T902} follows from the fact that conditioning reduces entropy, \eqref{T903} follows from the fact that the sources $S_1,\ldots,S_\alpha$ are mutually independent, and \eqref{T904} follows from the fact that the source $S_\alpha$ is memoryless. Moving $n\delta_n^{(\alpha)}$ to the right-hand side of the inequality completes the proof of \eqref{eq:ML}.

Next, let us we show that 
\begin{align}
&\sum_{U\in\Omega_L^{(1)}}\sum_{A\in \mathcal A_U}s_{\boldsymbol{\lambda}}(U,A)H(X_U|X_A)\notag\\
& \ge n\sum_{\alpha=1}^{m} f_\alpha(\boldsymbol{\lambda})H(S_\alpha)-n\sum_{\alpha=1}^{m}f_\alpha(\boldsymbol{\lambda})\delta_n^{(\alpha)}+\sum_{U\in\Omega_L^{(m)}}\sum_{A\in \mathcal A_U}s_{\boldsymbol{\lambda}}(U,A)H(X_U|S_1^n,\ldots,S_m^n,X_A)
\label{eq:induc}
\end{align}
for any $m=1,\ldots,L-N$. 

Consider a proof via an induction on $m$. First consider the base case with $m=1$. We have
\begin{eqnarray}
&&\sum_{U \in \Omega_L^{(1)}}\sum_{A \in \mathcal A(U)}s_{\boldsymbol{\lambda}}(U,A)H(X_U|X_A)\nonumber\\
&& \hspace{20pt} \ge \; \sum_{U \in \Omega_L^{(1)}}\sum_{A \in \mathcal A(U)}s_{\boldsymbol{\lambda}}(U,A)\left[nH(S_1)-n\delta_n^{(1)}+H(X_U|S_1^n,X_A)\right]\label{eq:c1}\\
&& \hspace{20pt} = \; nf_1(\boldsymbol{\lambda})H(S_1)-nf_1(\boldsymbol{\lambda})\delta_n^{(1)}+\sum_{U \in \Omega_L^{(1)}}\sum_{A \in \mathcal A_U}s_{\boldsymbol{\lambda}}(U,A)H(X_U|S_1^n,X_A)\label{eq:c2}
\end{eqnarray}
where \eqref{eq:c1} follows from \eqref{eq:ML} with $\alpha=1$. 

Next, assume that the inequality \eqref{eq:induc} holds for $m=k-1$ for some $k =2,\ldots,L-N$, i.e.,
\begin{align}
&\sum_{U \in \Omega_L^{(1)}}\sum_{A\in \mathcal{A}_U}s_{\boldsymbol{\lambda}}(U,A)H(X_U|X_A) \notag\\
& \ge n\sum_{\alpha=1}^{k-1}f_\alpha(\boldsymbol{\lambda})H(S_\alpha)-n\sum_{\alpha=1}^{k-1}f_\alpha(\boldsymbol{\lambda})\delta_n^{(\alpha)}+\sum_{V\in\Omega_L^{(k-1)}}\sum_{A\in \mathcal{A}_U}s_{\boldsymbol{\lambda}}(U,A)H(X_U|S_1^n,\ldots,S_{k-1}^n,X_A).
\label{eq:induc2}
\end{align}
We have
\begin{eqnarray}
\hspace{-20pt}&&\sum_{U \in \Omega_L^{(k-1)}}\sum_{A\in \mathcal{A}_U}s_{\boldsymbol{\lambda}}(U,A)H(X_U|S_1^n,\ldots,S_{k-1}^n,X_A)\nonumber\\
\hspace{-20pt}&& \hspace{20pt} \geq \; \sum_{U \in \Omega_L^{(k)}}\sum_{A\in \mathcal{A}_U}s_{k}(U,A)H(X_U|S_1^n,\ldots,S_{k-1}^n,X_A)\label{eq:c3}\\
\hspace{-20pt}&& \hspace{20pt} \geq \; \sum_{U \in \Omega_L^{(k)}}\sum_{A\in \mathcal{A}_U}s_{k}(U,A)\left[nH(S_{k})-n\delta_n^{(k)}+H(X_U|S_1^n,\ldots,S_k^n,X_A)\right]\label{eq:c4}\\
\hspace{-20pt}&& \hspace{20pt} \geq \; nf_k(\boldsymbol{\lambda})H(S_k)-nf_k(\boldsymbol{\lambda})\delta_n^{(k)}+\sum_{U \in \Omega_L^{(k)}}\sum_{A \in \mathcal{A}_U}s_{k}(U,A)H(X_U|S_1^n,\ldots,S_k^n,X_A)\label{eq:c5}
\end{eqnarray}
where \eqref{eq:c3} follows from \eqref{eq:CYZ2}, and \eqref{eq:c4} follows from \eqref{eq:ML} with $\alpha=k$. Substituting \eqref{eq:c5} into \eqref{eq:induc2} gives
\begin{align}
&\sum_{U\in\Omega_L^{(1)}}\sum_{A\in \mathcal A_U}s_{\boldsymbol{\lambda}}(U,A)H(X_U|X_A)\notag\\
& \ge n\sum_{\alpha=1}^{k} f_\alpha(\boldsymbol{\lambda})H(S_\alpha)-n\sum_{\alpha=1}^{k}f_\alpha(\boldsymbol{\lambda})\delta_n^{(\alpha)}+\sum_{U\in\Omega_L^{(k)}}\sum_{A\in \mathcal A_U}s_{\boldsymbol{\lambda}}(U,A)H(X_U|S_1^n,\ldots,S_k^n,X_A)
\label{eq:induc3}
\end{align}
i.e., the inequality \eqref{eq:induc} also holds for $m=k$. This completes the induction step and hence the proof of \eqref{eq:induc}.

Finally, note that for $\alpha=1$ the optimal solution for the linear program \eqref{eq:dual2} is \emph{unique} and is given by
\begin{equation}
c_{\boldsymbol{\lambda}}(\{l\})=\lambda_l, \quad \forall l=1,\ldots,L.
\end{equation}
We thus have
\begin{eqnarray}
n\left(\sum_{l=1}^L\lambda_lR_l\right) &=&\sum_{l=1}^Lc_{\boldsymbol{\lambda}}(\{l\})nR_l\\
& \geq & \sum_{l=1}^Lc_{\boldsymbol{\lambda}}(\{l\})(H(X_l)-n\epsilon) \label{Pf-SSMDC-1}\\
& = & \sum_{l=1}^Lc_{\boldsymbol{\lambda}}(\{l\})H(X_l)-nf_1(\boldsymbol{\lambda})\epsilon \label{Pf-SSMDC-2}\\
& = & \sum_{U \in \Omega_L^{(1)}}c_{\boldsymbol{\lambda}}(U)H(X_U)-nf_1(\boldsymbol{\lambda})\epsilon\\
& = & \sum_{U \in \Omega_L^{(1)}}\left[\sum_{A\in \mathcal A_U}s_{\boldsymbol{\lambda}}(U,A)\right]H(X_U)-nf_1(\boldsymbol{\lambda})\epsilon\\
& = & \sum_{U \in \Omega_L^{(1)}}\sum_{A\in \mathcal A_U}s_{\boldsymbol{\lambda}}(U,A)H(X_U)-nf_1(\boldsymbol{\lambda})\epsilon\\
&\ge &\sum_{U\in\Omega_L^{(1)}}\sum_{A\in \mathcal A_U}s_{\boldsymbol{\lambda}}(U,A)H(X_U|X_A)-nf_1(\boldsymbol{\lambda})\epsilon\label{Pf-SSMDC-3}\\
&\geq &  \left[n\sum_{\alpha=1}^{L-N}f_\alpha(\boldsymbol{\lambda})H(S_\alpha)-n\sum_{\alpha=1}^{L-N}f_\alpha(\boldsymbol{\lambda})\delta_n^{(\alpha)}+\right.\nonumber\\
&&\left.\sum_{U\in\Omega_L^{(L-N)}}\sum_{A\in \mathcal A_U}s_{\boldsymbol{\lambda}}(U,A)H(X_U|S_1^n,\ldots,S_{L-N}^n,X_A)\right]-nf_1(\boldsymbol{\lambda})\epsilon\label{Pf-SSMDC-4}\\
&\geq & n\sum_{\alpha=1}^{L-N} f_\alpha(\boldsymbol{\lambda})H(S_\alpha)-n\sum_{\alpha=1}^{L-N}f_\alpha(\boldsymbol{\lambda})\delta_n^{(\alpha)}-nf_1(\boldsymbol{\lambda})\epsilon\label{Pf-SSMDC-5}
\end{eqnarray}
where \eqref{Pf-SSMDC-1} follows from the rate constraint \eqref{eq:M-Rate}, \eqref{Pf-SSMDC-2} follows from the fact that $c_{\boldsymbol{\lambda}}^{(1)}$ is optimal so $f_1(\boldsymbol{\lambda}) = \sum_{l=1}^Lc_{\boldsymbol{\lambda}}(\{l\})$, \eqref{Pf-SSMDC-3} follows from the fact that conditioning reduce entropy, and \eqref{Pf-SSMDC-4} follows from \eqref{eq:induc} with $m=L-N$. Divide both sides of \eqref{Pf-SSMDC-5} by $n$ and let $n \rightarrow \infty$ and $\epsilon \rightarrow 0$. Note that $\delta_n^{(\alpha)}\rightarrow 0$ in the limit as $n \rightarrow \infty$ and $\epsilon \rightarrow 0$ for all $\alpha=1,\ldots,L-N$. We have thus proved that \eqref{eq:CJH} holds for any admissible rate tuple $(R_1,\ldots,R_L)$. This completes the proof of the theorem.
\end{proof}

\section{Concluding Remarks} \label{sec:Con}
SMDC is a classical model for coding over distributed storage. In this setting, a simple separate encoding strategy known as superposition coding was shown to be optimal in terms of achieving the minimum sum rate \cite{RYH-IT97} and the entire admissible rate region \cite{YZ-IT99} of the problem. The proofs utilized carefully constructed induction arguments, for which the classical subset entropy inequality of Han \cite{Han-IC78} played a key role. 

This paper includes two parts. In the first part the existing optimality proofs for classical SMDC were revisited, with a focus on their connections to subset entropy inequalities. First, a new sliding-window subset entropy inequality was introduced and then used to establish the optimality of superposition coding for achieving the minimum sum rate under a weaker source-reconstruction requirement. Second, a subset entropy inequality recently proved by Madiman and Tetali \cite{MT-IT10} was used to develop a new structural understanding to the proof of Yeung and Zhang \cite{YZ-IT99} on the optimality of superposition coding for achieving the entire admissible rate region. Building on the connections between classical SMDC and the subset entropy inequalities developed in the first part, in the second part the optimality of superposition coding was further extended to the cases where there is either an additional all-access encoder (SMDC-A) or an additional secrecy constraint (S-SMDC).

Finally, we mention here that an ``asymmetric" setting of the multilevel diversity coding problem was considered in the recent work \cite{MTD-IT10}, where the sources that need to be asymptotically perfectly reconstructed depend on, not only the cardinality, but the actual subset of the encoder outputs available at the decoder. Unlike the symmetrical setting considered in \cite{Roc-Thesis92,Yeu-IT95,RYH-IT97,YZ-IT99} and in this paper, as demonstrated in \cite{MTD-IT10} for the case with three encoders,  coding across different sources is generally needed to achieve the entire admissible rate region of the problem.

\appendix

\section{Proof of Theorem~\ref{thm:YZ2}}\label{app:A}
Consider a proof via an induction on the total number of encoders $L$. Fix $\boldsymbol{\lambda} \in (\mathbb{R}^+)^L$. Without loss of generality, let us assume that
\begin{equation}
\lambda_1 \geq \lambda_2 \geq \cdots \geq \lambda_L\geq 0.
\end{equation}

First consider the base case with $L=2$. In this case, the optimal solution to the linear program \eqref{eq:dual2} is unique and is given by
\begin{equation}
c_{\boldsymbol{\lambda}}(\{l\})=\lambda_l, \quad l=1,2 \quad \mbox{and} \quad c_{\boldsymbol{\lambda}}(\{1,2\})=\lambda_2.
\end{equation}
When $f_2(\boldsymbol{\lambda})=\lambda_2>0$, it is straightforward to verify that 
\begin{equation}
g_{\{1,2\}}(\{l\})=\lambda_l/\lambda_2, \quad l=1,2
\end{equation}
is a fractional cover of $(\{1,2\},\{\{1\},\{2\}\})$ and such that
\begin{align}
c_{\boldsymbol{\lambda}}(\{l\})=g_{\{1,2\}}(\{l\})c_{\boldsymbol{\lambda}}(\{1,2\}), \quad l=1,2.
\end{align}

Now, assume that the theorem holds for $L=N-1$ for some integer $N \geq 3$. Fix $\alpha \in \{2,\ldots,N\}$, and let $c_{\boldsymbol{\lambda}}^{(\alpha)}$ be an optimal solution to the linear program to \eqref{eq:dual2} with the optimal value $f_{\alpha}(\boldsymbol{\lambda})>0$. Next, we show that we can always find a collection of functions $\{g_U: U \in \Omega_L^{(\alpha)}\}$ for which each $g_U$ is a fractional cover of $(U,\mathcal{V}_U)$ and such that $c_{\boldsymbol{\lambda}}^{(\alpha-1)}=\{c_{\boldsymbol{\lambda}}(V): V \in \Omega_L^{(\alpha-1)}\}$ where $c_{\boldsymbol{\lambda}}(V)$ is given by \eqref{eq:YZ2} is an optimal solution to the linear program \eqref{eq:dual2} with $\alpha$ replaced by $\alpha-1$.

We shall consider the following three cases separately.

Case 1: $\lambda_1 \leq \frac{\lambda_2+\cdots+\lambda_N}{\alpha-1}$. In this case, it is sufficient to consider for any $U \in \Omega_N^{(\alpha)}$, the \emph{uniform} fractional cover
\begin{equation}
g_U(V)=\frac{1}{\alpha-1}, \quad \forall V \in  \mathcal{V}_U
\end{equation}
for the hypergraph $(U,\mathcal{V}_U)$
so we have
\begin{equation}
c_{\boldsymbol{\lambda}}(V)=\sum_{U \in \mathcal{U}_V}\frac{c_{\boldsymbol{\lambda}}(U)}{\alpha-1}, \quad \forall V \in \Omega_N^{(\alpha-1)}.
\end{equation}
By \cite[Eq.~(39)]{YZ-IT99}, $c_{\boldsymbol{\lambda}}^{(\alpha-1)}$ constructed as such is an optimal solution to the linear program \eqref{eq:dual2} with $\alpha$ replaced by $\alpha-1$.

Case 2: $\lambda_1 > \frac{\lambda_2+\cdots+\lambda_N}{\alpha-2}$. In this case, by \cite[Lemma~6]{YZ-IT99} $c_\alpha(U) >0$ implies that $U \ni 1$. Furthermore, by \cite[Lemma~8]{YZ-IT99} $\tilde{c}_{\boldsymbol{\lambda}}^{(\alpha-1)}=\{\tilde{c}_{\boldsymbol{\lambda}}(\tilde{U}):\tilde{U} \subseteq \tilde{\Omega}_{N-1}:=\{2,\ldots,N\}\}$ where
\begin{equation}
\tilde{c}_{\boldsymbol{\lambda}}(\tilde{U})=c_{\boldsymbol{\lambda}}(\{1\}\cup\tilde{U})
\end{equation}
is an optimal solution to the linear program
\begin{equation}
\begin{array}{rcl}
\max &&  \sum_{\tilde{U} \in \tilde{\Omega}_{N-1}^{(\alpha-1)}} \tilde{c}_{\boldsymbol{\lambda}}(\tilde{U})\\
\mbox{subject to} &&  \sum_{\tilde{U} \in \tilde{\Omega}_{N-1}^{(\alpha-1)}, \tilde{U} \ni l} \tilde{c}_{\boldsymbol{\lambda}}(\tilde{U}) \leq \lambda_l, \quad \forall l=2,\ldots,N \\
&& \tilde{c}_{\boldsymbol{\lambda}}(\tilde{U}) \geq 0, \quad \forall \tilde{U} \in \tilde{\Omega}_{N-1}^{(\alpha-1)}
\end{array}
\label{eq:LP3}
\end{equation}
with the optimal solution $\tilde{f}_{\alpha-1}(\boldsymbol{\lambda})=f_\alpha(\boldsymbol{\lambda})>0$. Thus, by the induction assumption there exists a collection of functions $\{\tilde{g}_{\tilde{U}}:\tilde{U} \in \tilde{\Omega}_{N-1}^{(\alpha-1)}\}$ such that each $\tilde{g}_{\tilde{U}}$ is a fractional cover of $(\tilde{U},\tilde{\mathcal{V}}_{\tilde{U}})$ and $\tilde{c}_{\boldsymbol{\lambda}}^{(\alpha-2)}=\{\tilde{c}_{\boldsymbol{\lambda}}(\tilde{V}):\tilde{V} \in \tilde{\Omega}_{N-1}^{(\alpha-2)}\}$ where
\begin{equation}
\tilde{c}_{\boldsymbol{\lambda}}(\tilde{V}) := \sum_{\tilde{U} \in \tilde{\mathcal{U}}_{\tilde{V}}}
\tilde{c}_{\boldsymbol{\lambda}}(\tilde{U})\tilde{g}_{\tilde{U}}(\tilde{V})
\end{equation}
is an optimal solution to the linear program
\begin{equation}
\begin{array}{rcl}
\max &&  \sum_{\tilde{V} \in \tilde{\Omega}_{N-1}^{(\alpha-2)}} \tilde{c}_{\boldsymbol{\lambda}}(\tilde{V})\\
\mbox{subject to} &&  \sum_{\tilde{V} \in \tilde{\Omega}_{N-1}^{(\alpha-2)}, \tilde{V} \ni l} \tilde{c}_{\boldsymbol{\lambda}}(\tilde{V}) \leq \lambda_l, \quad \forall l=2,\ldots,N \\
&& \tilde{c}_{\boldsymbol{\lambda}}(\tilde{V}) \geq 0, \quad \forall \tilde{V} \in \tilde{\Omega}_{N-1}^{(\alpha-2)}.
\end{array}
\label{eq:LP4}
\end{equation}

For any $U \in \Omega_{N}^{(\alpha)}$ such that $U \ni 1$, let $\tilde{U} =U \setminus \{1\}$, and let
\begin{equation}
g_U(V) := \left\{
\begin{array}{ll}
\tilde{g}_{\tilde{U}}(\tilde{V}), & \mbox{if $V=\{1\}\cup\tilde{V}$ for some $\tilde{V} \in \tilde{\mathcal{V}}_{\tilde{U}}$}\\
0, & \mbox{otherwise}.
\end{array}
\right.
\end{equation}
For any $U \in \Omega_{N}^{(\alpha)}$ such that $1 \notin U$, let us choose $g_U$ to be an \emph{arbitrary} fractional cover of $(U,\mathcal{V}_U)$. Then, for any $V \in \Omega_N^{(\alpha-1)}$ such that $V \ni 1$ we have
\begin{eqnarray}
c_{\boldsymbol{\lambda}}(V) &=& \sum_{U \in \mathcal{U}_V}c_{\boldsymbol{\lambda}}(U)g_U(V)\\
&=& \sum_{\tilde{U} \in \tilde{\mathcal{U}}_{\tilde{V}}}c_{\boldsymbol{\lambda}}(\{1\}\cup\tilde{U})\tilde{g}_{\tilde{U}}(\tilde{V})\\
&=& \sum_{\tilde{U} \in \tilde{\mathcal{U}}_{\tilde{V}}}\tilde{c}_{\boldsymbol{\lambda}}(\tilde{U})\tilde{g}_{\tilde{U}}(\tilde{V})\\
&=& \tilde{c}_{\boldsymbol{\lambda}}(\tilde{V})
\end{eqnarray}
where $\tilde{V}=V\setminus \{1\}$, and for any $V \in \Omega_N^{(\alpha-1)}$ such that $1 \notin V$
\begin{equation}
c_{\boldsymbol{\lambda}}(V) = \sum_{U \in \mathcal{U}_V}c_{\boldsymbol{\lambda}}(U)g_U(V) =0.
\end{equation}
By \cite[Eq.~(46)]{YZ-IT99}, $c_{\boldsymbol{\lambda}}^{(\alpha-1)}$ constructed as such is an optimal solution to the linear program \eqref{eq:dual2} with $\alpha$ replaced by $\alpha-1$. It remains to show that $g_U$ is a fractional cover of $(U,\mathcal{V}_U)$ for any $U \in \Omega_{N}^{(\alpha)}$ such that $U \ni 1$.

Fix $U \in \Omega_{N}^{(\alpha)}$ such that $U \ni 1$. For any $i \in U \setminus \{1\}$, we have
\begin{equation}
\sum_{\{V \in \mathcal{V}_U: V \ni i\}}g_U(V) = \sum_{\{V \in \mathcal{V}_U: V \supseteq \{1,i\}\}}g_U(V)
= \sum_{\{\tilde{V} \in \tilde{\mathcal{V}}_{\tilde{U}}: \tilde{V} \ni i\}}\tilde{g}_{\tilde{U}}(\tilde{V}) \geq 1
\end{equation}
and 
\begin{equation}
\sum_{\{V \in \mathcal{V}_U: V \ni 1\}}g_U(V) \geq \sum_{\{V \in \mathcal{V}_U: V \supseteq \{1,i\}\}}g_U(V)
\geq 1.
\end{equation}
This completes the proof of Case 2.

Case 3: $\frac{\lambda_2+\cdots+\lambda_N}{\alpha-1} <\lambda_1 \leq \frac{\lambda_2+\cdots+\lambda_N}{\alpha-2}$.
In this case, we shall need the following notations. For any $U \in \Omega_N^{(\alpha)}$ and $\tau \in \{1,\ldots,\alpha\}$, denote by $a_U(\tau)$ the \emph{smallest} positive integer $l$ such that 
\begin{equation}
|\{1,\ldots,l\}\cap U|=\tau.
\end{equation}
Let
\begin{equation}
W_\tau(U) := U\setminus \{a_U(\tau)\}
\end{equation}
so $W_\tau(U) \in \Omega_N^{(\alpha-1)}$. For each $U \in \Omega_N^{(\alpha)}$, $m \in \{2,\ldots,\alpha\}$, and $\tau \in \{m,\ldots,\alpha\}$, let $\xi_{U,m,\tau}: \Omega_N^{(\alpha-1)}\rightarrow \mathbb{R}^+$ where
\begin{align}
\xi_{U,m,\tau}(V) &:=\left\{
\begin{array}{ll}
\frac{b_{m-1}^{(\alpha)}-b_{m}^{(\alpha)}}{f_\alpha(\boldsymbol{\lambda})}, & \mbox{if $V=W_\tau(U)$}\\
0, & \mbox{otherwise}
\end{array}
\right.\\
b_l^{(\alpha)} & :=  \lambda_l-\tilde{\lambda}_l\\
\mbox{and} \quad \tilde{\lambda}_l & := \sum_{\{U\in \Omega_N^{(\alpha)},U \ni l\}}c_{\boldsymbol{\lambda}}(U), \quad \forall l=1,\ldots,L.
\end{align}
Let
\begin{equation}
\beta := \sum_{m=2}^{\alpha-1}(b_1^{(\alpha)}-b_m^{(\alpha)}).
\end{equation}
Consider the collection of functions $\{g_U: U \in \Omega_N^{(\alpha)}\}$ where
\begin{equation}
g_U(V):=\left(1-\frac{\beta}{f_\alpha(\boldsymbol{\lambda})}\right)\frac{1}{\alpha-1}+\sum_{m=2}^{\alpha}\sum_{\tau=m}^{\alpha}\xi_{U,m,\tau}(V), \quad \forall V \in \mathcal{V}_U.
\label{C3}
\end{equation}
This gives
\begin{equation}
c_{\boldsymbol{\lambda}}(V) = \left(1-\frac{\beta}{f_\alpha(\boldsymbol{\lambda})}\right)\sum_{U \in \mathcal{U}_V}\frac{c_{\boldsymbol{\lambda}}(U)}{\alpha-1}+\sum_{U \in \mathcal{U}_V}\sum_{m=2}^{\alpha}\sum_{\tau=m}^{\alpha}\xi_{U,m,\tau}(V)c_{\boldsymbol{\lambda}}(U), \quad \forall V \in \Omega_N^{(\alpha-1)}.
\end{equation}
By \cite[Eq.~(55)]{YZ-IT99}, $c_{\boldsymbol{\lambda}}^{(\alpha-1)}$ constructed as such is an optimal solution to the linear program \eqref{eq:dual2} with $\alpha$ replaced by $\alpha-1$. It remains to show that $g_U$ is a fractional cover of $(U,\mathcal{V}_U)$ for any $U \in \Omega_{N}^{(\alpha)}$

Note that for any $i \in U$,
\begin{equation}
\sum_{\{V \in \mathcal{V}_U, V \ni i\}}\left(1-\frac{\beta}{f_\alpha(\boldsymbol{\lambda})}\right)\frac{1}{\alpha-1}=1-\frac{\beta}{f_\alpha(\boldsymbol{\lambda})}
\label{C3-1}
\end{equation}
and
\begin{eqnarray}
\sum_{\{V \in \mathcal{V}_U, V \ni i\}}\sum_{m=2}^{\alpha}\sum_{\tau=m}^{\alpha}\xi_{U,m,\tau}(V) &=&
\sum_{m=2}^{\alpha}\sum_{\tau=m}^{\alpha}\left(\sum_{\{V \in \mathcal{V}_U, V \ni i\}}\xi_{U,m,\tau}(V)\right)\\
&= &\sum_{m=2}^{\alpha}\sum_{\tau=m}^{\alpha}\frac{b_{m-1}^{(\alpha)}-b_{m}^{(\alpha)}}{f_\alpha(\boldsymbol{\lambda})}1_{\{a_U(\tau) \neq i\}}\\
&= &\sum_{m=2}^{\alpha}\frac{b_{m-1}^{(\alpha)}-b_{m}^{(\alpha)}}{f_\alpha(\boldsymbol{\lambda})}\left(\sum_{\tau=m}^{\alpha}1_{\{a_U(\tau) \neq i\}}\right)\\
&\geq &\sum_{m=2}^{\alpha}\frac{b_{m-1}^{(\alpha)}-b_{m}^{(\alpha)}}{f_\alpha(\boldsymbol{\lambda})}(\alpha-m)\\
&=& \frac{\beta}{f_\alpha(\boldsymbol{\lambda})}
\label{C3-2}
\end{eqnarray}
where \eqref{C3-2} follows from \cite[Eq.~(66)]{YZ-IT99}. Combing \eqref{C3-1} and \eqref{C3-2} gives
\begin{equation}
\sum_{\{V \in \mathcal{V}_U, V \ni i\}}g_U(V) \geq 1-\frac{\beta}{f_\alpha(\boldsymbol{\lambda})}+\frac{\beta}{f_\alpha(\boldsymbol{\lambda})}=1.
\end{equation}
We thus conclude that $g_U$ as defined in \eqref{C3} is indeed a fractional cover of $(U,\mathcal{V}_U)$ for any $U \in \Omega_N^{(\alpha)}$. This completes the proof of Case 3.

\section*{Acknowledgement}
Tie Liu would like to thank Dr. Jihong Chen for discussions that have inspired some ideas of the paper.

\end{document}